\documentclass[AMA,STIX1COL]{WileyNJD-v2}
\pdfoutput=1

\articletype{Research Article}%

\received{26 January 2022}
\revised{ 2022}
\accepted{ 2022}

\newcommand{\hide}[1]{}

\usepackage{bm}
\usepackage{textgreek}
\usepackage{units}
\usepackage{xspace}
\usepackage{afterpage}
\usepackage[capitalize]{cleveref}             

\usepackage{mathtools}
\usepackage{derivative}
\usepackage{threeparttable}
\usepackage{lscape}
\usepackage{float}
\usepackage{booktabs}
\usepackage{threeparttable}
\usepackage{caption}
\usepackage{pgfplots}
\usepackage{subcaption}
\usepackage[short]{optidef}

\usetikzlibrary{plotmarks,shapes}
\usetikzlibrary{arrows.meta, arrows}
\usepgfplotslibrary{patchplots}
\pgfplotsset{try min ticks=5}
\pgfplotsset{scaled y ticks=false}
\pgfplotsset{compat=newest}

\newlength\figHsmall
\newlength\figWsmall
\newlength\figHbig
\newlength\figWbig
\graphicspath{{figures/}}

\newcommand*{\ie}{i.e.}

\newcommand*{\eg}{e.g.,}

\makeatletter
\NAT@longnamesfalse
\makeatother

\begin{document}

\title{Safe Hierarchical Model Predictive Control and Planning for Autonomous Systems}

\author[1]{Markus K\"ogel}
\author[1]{Mohamed Ibrahim}
\author[2]{Christian Kallies}
\author[3]{Rolf Findeisen}
\authormark{K\"ogel, Ibrahim, Kallies, Findeisen}

\address[1]{\orgdiv{Laboratory for Systems Theory and Automatic Control},
\orgname{Otto von Guericke University}, \orgaddress{\city{Magdeburg},
\country{Germany}}}

\address[2]{\orgdiv{Institute of Flight Guidance}, \orgname{\newline German Aerospace Center (DLR)}, \orgaddress{\city{Braunschweig},
\country{Germany}}}

\address[3]{\orgdiv{Control and Cyber-Physical Systems Laboratory},
\orgname{Technical University of Darmstadt}, \orgaddress{
\country{Germany}}}

\corres{Rolf Findeisen, \email{rolf.findeisen@tu-darmstadt.de}, \orgaddress{\city{\newline TU Darmstadt}, \country{Germany}}}


\abstract[Abstract]{Planning and control for autonomous vehicles usually are hierarchical separated. However, increasing performance demands and operating in highly dynamic environments requires an frequent re-evaluation of the planning and tight integration of control and planning to guarantee safety. We propose an integrated hierarchical predictive control and planning approach to tackle this challenge. Planner and controller are based on the repeated solution of moving horizon optimal control problems. The planner can choose different low-layer controller modes for increased flexibility and performance instead of using a single controller with a large safety margin for collision avoidance under uncertainty. Planning is based on simplified system dynamics and safety, yet flexible operation is ensured by constraint tightening based on a mixed-integer linear programming formulation. A cyclic horizon tube-based model predictive controller guarantees constraint satisfaction for different control modes and disturbances. Examples of such modes are a slow-speed movement with high precision and fast-speed movements with large uncertainty bounds. Allowing for different control modes reduces the conservatism, while the hierarchical decomposition of the problem reduces the computational cost and enables real-time implementation. We derive conditions for recursive feasibility to ensure constraint satisfaction and obstacle avoidance to guarantee safety and ensure compatibility between the layers and modes. Simulation results illustrate the efficiency and applicability of the proposed hierarchical strategy.}



\maketitle

\section{Introduction}\label{sec.intro}
Autonomous vehicles, such as drones, mobile robots, or autonomous transportation systems, are by now used in a wide range of applications, spanning from geo-surveillance, agricultural application, logistics, or search and rescue operations~\cite{paden2016survey,schwarting2018planning}. Often, the autonomous vehicle has to achieve a task, like to move from a start to a goal position, while avoiding obstacles in a dynamically changing environment, compare  Figure~\ref{fig:main_kay}.
\begin{figure}[htb]
	\begin{center}
		\includegraphics[width=0.7\columnwidth]{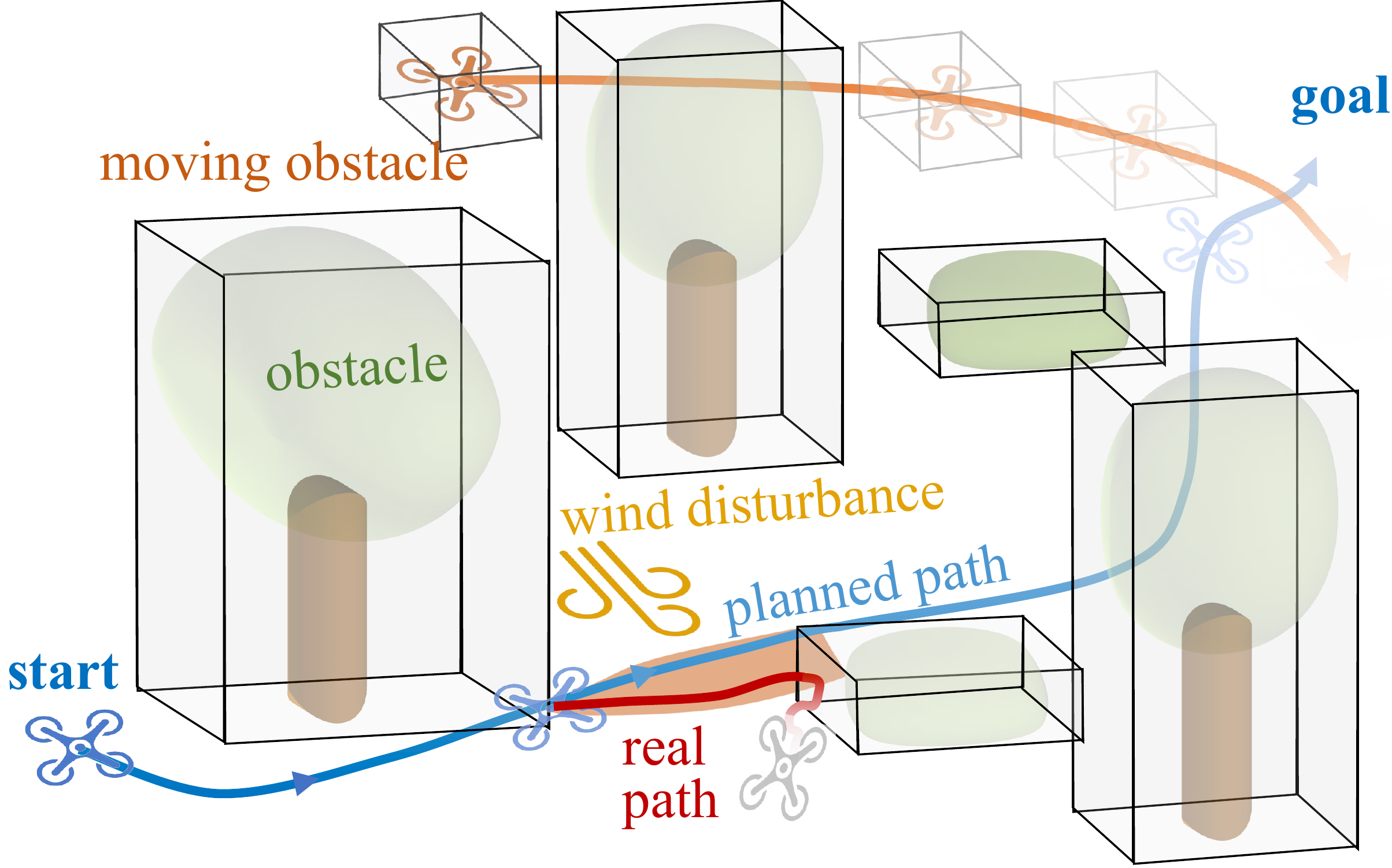}
		\caption{Planning and control for autonomous vehicles in uncertain, dynamically changing environments, e.g., moving obstacles or wind, with collision avoidance.} \vspace*{0.0cm}
		\label{fig:main_kay}
	\end{center}	
\end{figure}

Safe operation, i.e., collision avoidance, under all circumstances is challenging~\cite{dadkhah2012survey,kingston2018sampling,ibrahim2020real}. Typically, the task is hierarchically separated into a planning task performed once or repeatedly on a slow time scale. It provides a reference or path to a lower-layer control layer that operates on a fast time scale to counteract disturbances, model uncertainty, or react to (fast) environmental changes, c.f. Figure \ref{fig:blkdiag} (a), see also. ~\cite{rosolia2020unified,pant2021co,broderick2014optimal,ibrahim2019hierarchical,yin2020optimization,ibrahim2019improved,koeln2018two,vermillion2014stable}

Many planning approaches for autonomous vehicles exist, see, e.g., \cite{paden2016survey,goerzen2010survey,quan2020survey,mac2016heuristic} and references their-in. They, for example, exploit deterministic or heuristic search strategies, or reformulate the problem as a mathematical optimization problem. However, most of them do not allow considering detailed vehicle dynamics, environmental conditions, or disturbances explicitly.  
\begin{figure}[htb]
	\begin{center}
    \begin{subfigure}[t]{.49\textwidth}
		\centering
		\includegraphics[width=\textwidth]{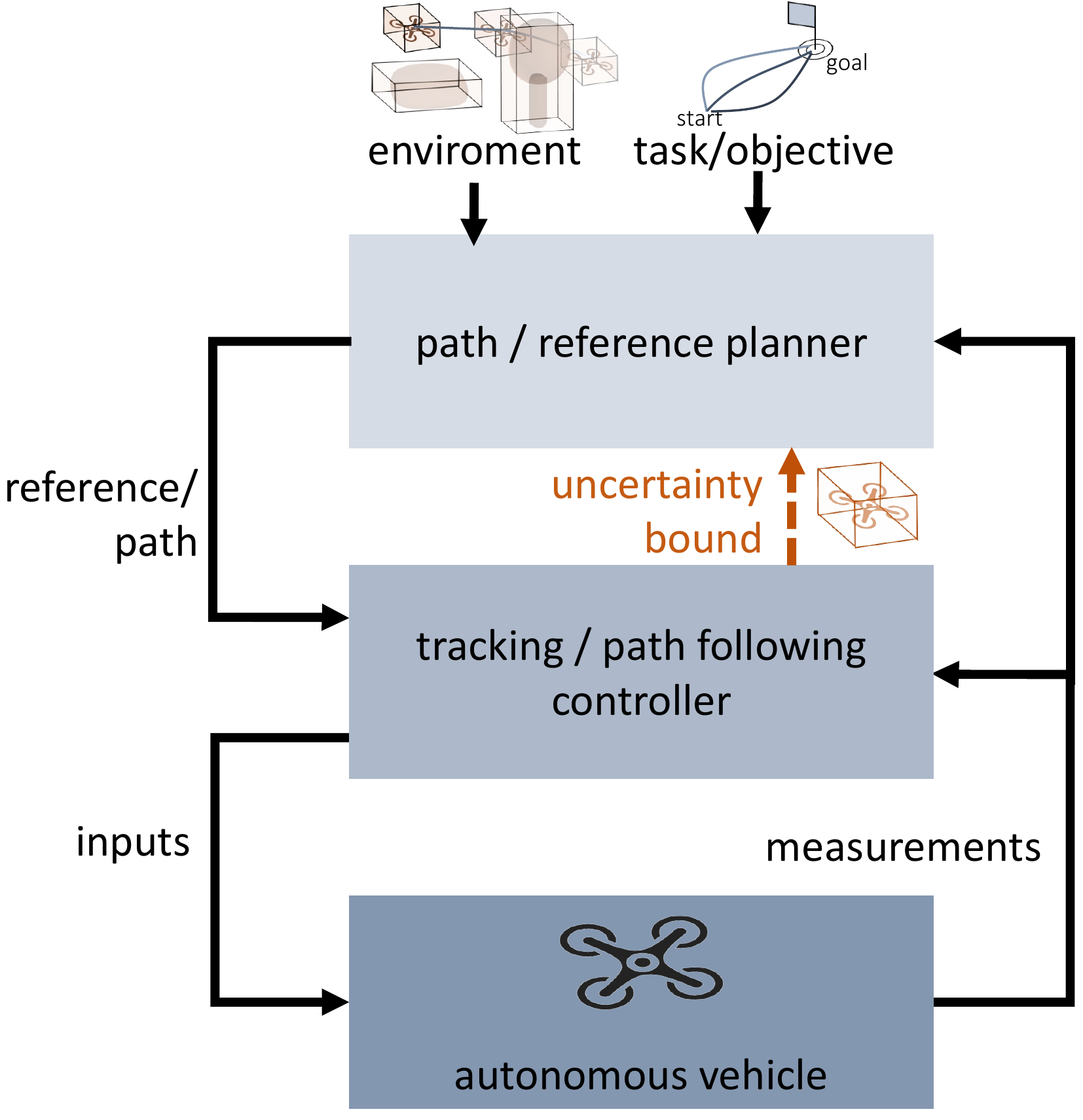}
		\captionsetup{font=normalsize,labelfont={bf,sf}}
		\caption{Typical decomposition of planning and control.}
		\label{fig:pf1}
	\end{subfigure}%
    \hfill
	\begin{subfigure}[t]{.49\textwidth}
		\centering
		\includegraphics[width=\textwidth]{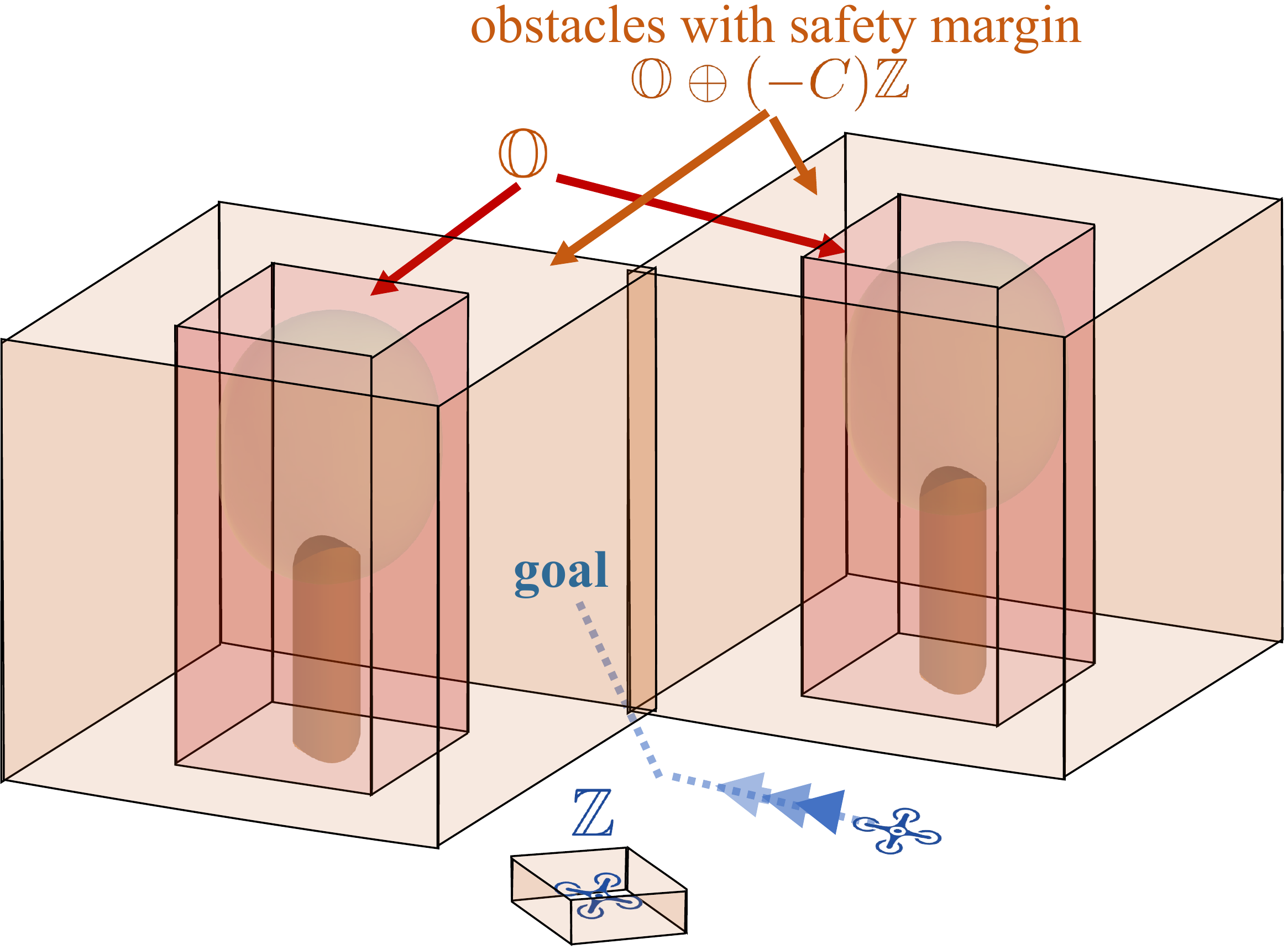}
		\captionsetup{font=normalsize,labelfont={bf,sf}}
		\caption{Safe planning by constraint tightening/constraint back off considering a single, fixed tracking control mode.}
		\label{fig:nosol}
	\end{subfigure}
		\caption{(a) Hierarchical separation of the planning and control layers; (b) Safe planning with a single tracking control mode.} \vspace*{0.0cm}
		\label{fig:blkdiag}
	\end{center}	
\end{figure}
Counteracting significant uncertainties, e.g., due to changing environmental conditions like wind, or operating in a dynamically evolving environment, requires frequent replanning and tightly integrating control and planning. This has attracted significant research over the past years: A multi-rate hierarchical approach consisting of three control layers operating at different sampling times is presented in~\cite{rosolia2020unified}. A Markov decision problem is exploited on the highest planning layer, which feeds into the planning/reference generation layer, exploiting a model predictive control formulation. At the lowest layer, a tracking controller is considered that uses a barrier function approach. 
In~\cite{pant2021co} an approach to co-design the planner and control is proposed, which allows considering vehicle dynamics and kinematic constraints. The tracking error of the controller is directly considered in the planner. A  multi-layer control framework was presented in~\citep{yin2020optimization}, where the optimization-based reference planner used a low fidelity model. At the same time, a feedback controller tracks the planned trajectory with a specific error bound. The state and input constraints are taken into account in the planning layer. The tracking control law and its error bounds are parameterized and computed offline using sum-of-squares programming. The approach allows maintaining safety by taking the tracking error bound into account to achieve maximum permissiveness of the planning layer. A two-layer  model predictive control framework is presented in~\citep{koeln2018two}, where the upper-layer controller determines the state trajectories over the entire mission while guaranteeing the satisfaction of the system constraint during operation. The lower-layer controller modifies the determined trajectories to improve reference tracking. In~\cite{vermillion2014stable}, a stable hierarchical model predictive control (MPC) is introduced using an inner loop reference model and  contracting constraint sets for guaranteeing the overall stability. The inner-loop controller tracks the output of a prescribed reference model.
\begin{figure}[htb]
	\begin{center}
    \begin{subfigure}[t]{.49\textwidth}
		\centering
		\includegraphics[width=\textwidth]{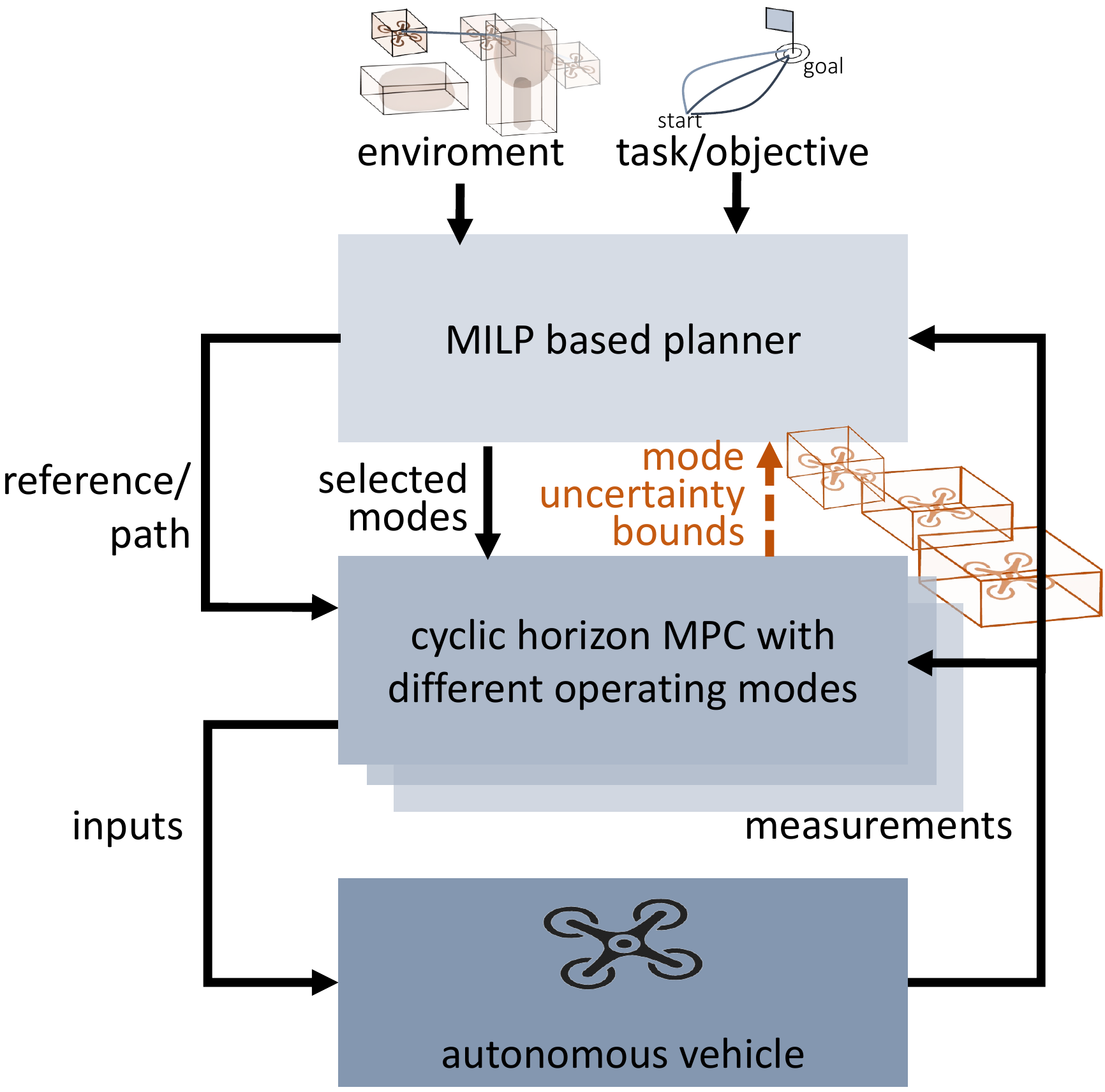}
		\captionsetup{font=normalsize,labelfont={bf,sf}}
		\caption{Proposed moving horizon planning and control exploiting multiple controller operation regions/modes.}
		\label{fig:MPCplanningcontrol}
	\end{subfigure}%
    \hfill
	\begin{subfigure}[t]{.49\textwidth}
		\centering
		\includegraphics[width=\textwidth]{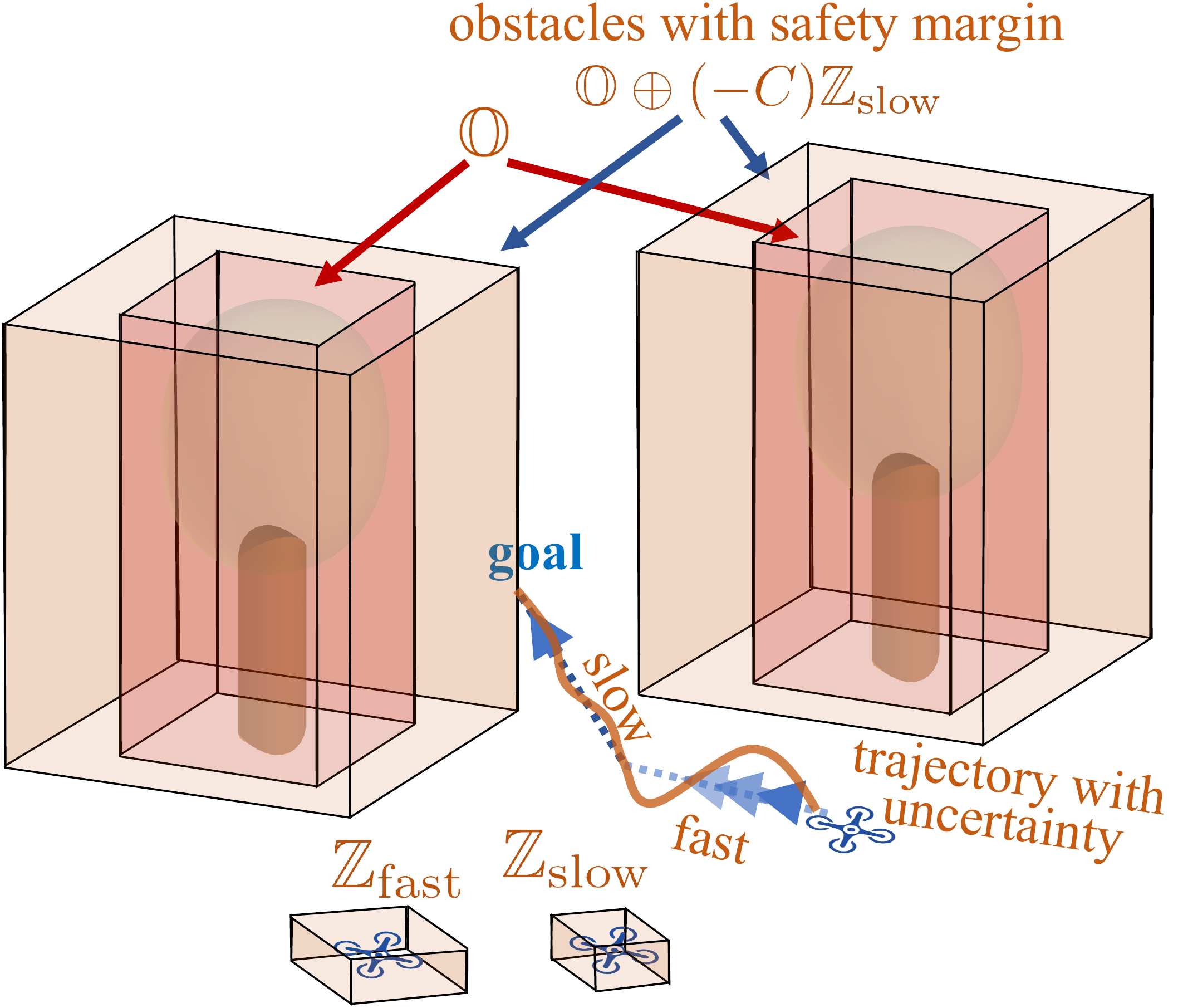}
		\captionsetup{font=normalsize,labelfont={bf,sf}}
		\caption{Reduced conservatism due to different tracking controller operation regions/modes.}
		\label{fig:pf2new}
	\end{subfigure}
		\caption{Proposed planning and control scheme with multiple tracking controller operation regions/modes.} \vspace*{0.0cm}
		\label{fig:multiple_modes}
	\end{center}	
\end{figure}

We propose a tightly integrated hierarchical planning and control approach to reduce conservatism while being computational feasible. Both planner and controller are based on the repeated solution of moving horizon optimal control problems. In particular, both layers are based on  robust MPC~\cite{grune2017nonlinear,findeisen2002introduction,rawlings2017book,chisci2001systems,mayne2005robust} formulations. The planning- and the lower-layer controller agree on  ``contracts''  (safety corridors), guaranteeing consistency and compatibility between the layers. Consequently, this ensures robust satisfaction of the constraints such as collision avoidance even in the presence of disturbances. The planner can choose different operating modes, corresponding to different accuracy's that the lower-layer controller can achieve in the closed-loop. In comparison, existing hierarchical formulations often rely on a single yet conservative tracking controller mode, c.f., Fig~\ref{fig:blkdiag}. 

A specific example of such modes is a slow-speed movement with high precision and a fast-speed movement with large uncertainty bounds, c.f.,~\ref{fig:pf2new}. As illustrated in Fig.~\ref{fig:MPCplanningcontrol} (b), the two control modes enable the planner to provide a collision-free path by switching online between different velocity ranges and corresponding controller parametrization. 

We propose a moving horizon formulation for the planning layer, resulting in an efficiently solvable mixed-integer linear programming (MILP). The tracking control for the vehicle is achieved by a cyclic horizon MPC~\cite{chisci2001systems,kogel2013stability}, which operates with a high sampling frequency. The tracking controller provides safety bounds for each operation mode, despite disturbances. We derive conditions for recursive feasibility to ensure constraint satisfaction and obstacle avoidance. Simulation results illustrate the efficiency and applicability of the proposed hierarchical strategy. Especially, it is shown how different modes reduce the conservatism, while the decomposition of the control problem limits the computational cost and enables real-time implementation while providing guarantees. 

The presented results are based on the results presented in \cite{ibrahim2020contract}. They are expanded with respect to less conservative conditions and formulations and the overall approach is evaluated  considering a quadcopter operating in a 3D environment.

The main contributions of this work are threefold:
\begin{enumerate}
    \item First, we outline a planning approach operating on a moving horizon, which allows us to consider different tracking controller modes. We show how this approach can be reformulated as a MILP, which can be efficiently solved in real-time. 
    \item Second, we design a low-layer cyclic horizon tube MPC controller, which provides safety bounds, i.e., tubes, for the different modes, e.g., velocities. This controller guarantees in combination with the planning layer constraint satisfaction and collision avoidance.
    \item Finally, we demonstrate our approach, considering a quadcopter operating in a 3D environment.
\end{enumerate}

The remainder of the paper is organized as follows. Section~\ref{sec.promlem} presents the problem setup. Section~\ref{sec.plan_control} outlines the hierarchy planning and control scheme, where Section~\ref{sec.plan} presents the MILP moving horizon planning problem.
Section~\ref{sec:PF} describes the low-layer cyclic horizon tube MPC and Section~\ref{Scheduler} outlines the switching between the operating modes.
Section~\ref{sec:ex_res}  presents the simulation example and results for an quadcopter, while Section~\ref{sec.Conclusion} recapitulates our findings, conclusions, and outlooks. \\
 
\noindent\emph{Notation} For two given sets $\mathbb{X},\mathbb{Y} \subset \mathbb{R}^n$, the Minkowski set sum $\oplus$ and the Minkowski set difference $\ominus$ are defined by:
$\mathbb{X}  \oplus  \mathbb{Y} \triangleq \{x+y|x \in \mathbb{X},y \in \mathbb{Y}\}$, 
$\mathbb{X}  \ominus \mathbb{Y} \triangleq \{x|x\oplus\mathbb{Y}\subseteq \mathbb{X}$. We use $\operatorname{rem}$ to denote the remainder function of Euclidean division, $\mathbf{1}$ denotes a column vector with $1$ in each entry, and 
$\|x\|_Q^2=x^TQx$. 

\section{Problem Formulation}\label{sec.promlem}
We consider the control of an autonomous vehicle, which should move (drive, fly, etc.) from a starting   point, $x(0)$, to  a goal, $x_\text{goal}$, while avoiding obstacles and satisfying constraints despite uncertainties and disturbances, see Fig.~\ref{fig:main_kay}.
We assume that the vehicle dynamics are subject to unknown but bounded additive disturbances and that they are governed by
\begin{subequations}
	\label{system}	\begin{align}
	x(k+1)=\, & A x(k) + Bu(k)+w(k),\label{eq:x}  \\
	y(k)=\, & C x(k), \\
	x(k)\in \,& \mathbb{X},\,\,u(k)\in \mathbb{U}.\label{state_input_constraints}
	\end{align}\end{subequations}
Here $x(k)\in \mathbb{R}^{n}$ is the vehicle's state,	$u(k)\in \mathbb{R}^{m}$ is the  applied control input, and the output $y(k)\in \mathbb{R}^{p}$, 
while $w(k)$ is an unknown, but bounded disturbance. The  state $x(k)$  and  the input $u(k)$ need to satisfy constraints: they are restricted to the sets $\mathbb{X}   $ and $\mathbb{U} $, which  are both closed and convex. \\

\noindent\emph{Obstacle avoidance:}
Beside the constraints~\eqref{state_input_constraints} on the state $x(k)$ and input $u(k)$ of the autonomous vehicle, we want to achieve obstacle avoidance, which we formulate in terms of the output $y(k)$ (e.g. the position). 
We assume that there are $H$ obstacles  and that each obstacle is modeled as a bounded set of the form $\mathbb{O}_\ell=\lbrace y\vert  E_\ell y< f_\ell \rbrace$, where $E_\ell\in\mathbb{R}^{q_\ell\times p}$. So, $\mathbb{O}_\ell$ is  the interior of a  convex, compact polytope. In the most simple case (box obstacles) $E_\ell=[I\,\,-I]$. 
Consequently, to avoid  that the vehicle  collides or touches the   obstacles we require that:
\begin{align}\label{eq:all_obstacles}
y (k)\notin \mathbb{O},\,\,
\mathbb{O}=\lbrace E_\ell y(k)< f_\ell,\,\,l=1,\ldots,H \rbrace, 
\end{align}
where $\mathbb{O}$ is the collection of all $H$ obstacles. Clearly, the set of admissible output/position $y (k)\notin \mathbb{O}$, defined in \eqref{eq:all_obstacles}, is  non-convex, but contains its boundaries. One can also formulate it as   
\begin{align}\label{obstacle_constraints}
\forall \ell\in \lbrace 1,\ldots,H\rbrace ,\,\,\exists a\in  \lbrace 1,...,q_\ell\rbrace:\ E_{\ell,a} y(k)\geq f_{\ell,a}.
\end{align} 
This formulation allows to use a MILP framework with the so-called big-M approach, which enables handling the obstacles systematically, see Section~\ref{Scheduler}. \\

\noindent\emph{Disturbance bounds and operating regions/modes:} 
The disturbance $w(k)$, and its bounds, may depend in parts also on the vehicle's state $x(k)$ and/or the applied control input $u(k)$. This might, for example be due to model uncertainties, which are captured as disturbances. For example, if the vehicle is operated at a lower speed, then the (worst case) disturbance might be smaller then at high speed. 
 
Specifically, we consider discrete sets of disturbance bounds for different low-layer/tracking controllers operating modes, i.e., we assume that there are  $N_w$ low-layer controller operating modes:
\begin{assumption}(Operating region/mode dependent disturbance bounds) \label{as:2}~
There are $N_w$ low-layer controller operating regions/modes defined by the state and input sets $U_i$ and $X_i$ such that if  $x(k)\in \mathbb{X}_i\subseteq \mathbb{X}$ and $u(k)\in \mathbb{U}_i\subseteq \mathbb{U}$, the uncertainty can be bounded by  $w(k) \in \mathbb{W}_i$, where the sets $\mathbb{X}_i$,  $\mathbb{U}_i$  are convex and closed polytopes and $\mathbb{W}_i$ is a convex and compact polytope.
\end{assumption} 
Note that the sets $ \mathbb{X}_i$ (and $ \mathbb{U}_i$) may overlap, i.e. one can have that $\mathbb{X}_i\cap \mathbb{X}_j\neq \emptyset$ for $i\neq j$. 

\section{Robust Hierarchical Moving Horizon Planning and Control} \label{sec.plan_control}
We focus on a hierarchical two-layer planning and control decomposition as depicted in  Fig. \ref{fig:MPCplanningcontrol}, to achieve  safe and collision-free motion of the autonomous vehicle to the goal. Both layers utilize (robust) MPC formulations for constraint satisfaction under uncertainties.

To provide guarantees despite the hierarchical decomposition of the problem in a planning and control layer, we utilize the concept of  ``contracts'', inspired by \cite{bathge2018contract,lucia2015contract,blasi2018distributed}. Loosely speaking, a contract specifies the achievable precision in terms of a bound on the disturbance the lower-layer tracking controller can achieve for a specific operating region/mode $i$ - a particular set of states and inputs. In the proposed approach, the moving horizon planner provides a reference path and selects the operating region the controller should operate in. To calculate a safe passage - satisfy the constraints \eqref{state_input_constraints}  and to avoid all obstacles \eqref{obstacle_constraints} - the planner takes the uncertainty bounds corresponding to the different operating regions directly into account. 

In detail the upper-layer planner calculates the reference based on a model with simplified dynamics of the form
\begin{align}
x_p(k_p+1) = A_px_p(k_p) + B_pu_p(k_p),\label{eq:xp} 
\end{align}
where $k_p$ is the planning time index.
The generated reference path and measurements of the vehicle's state  are used  by the  lower-layer controller to calculate the control input $u(k)$, see Fig. \ref{fig:blkdiag}. 
 The  lower-layer control loop aims to efficiently counteract the disturbances, to ensure that the vehicle follows the planned reference with   a specified accuracy,  to  satisfy the contract,  and 
to guarantee constraint satisfaction. As the obstacles are handled by the planner, the tracking controller does not need to consider them, which enables a fast and efficient implementation as non-convex constraints are avoided. 
  
We consider that the planning is done slowly, all $M\ (>1)$ time steps. So we have  
\begin{align}
k=k_p\cdot M,\quad
A_p=A^M,\quad B_p=\sum_{i=0}^{M-1}A^iB.\label{eq:apa}
\end{align} 
We assume that the the real dynamics \eqref{eq:x} and the "planning dynamics" \eqref{eq:xp} satisfy:
\begin{assumption}(Controllability of the planning and control dynamics) \label{as:1}
	The pairs $(A, B)$ and $(A_p, B_p)$ are controllable.
\end{assumption}
Note that if $(A_p,B_p)$ is controllable, then also $(A,B)$. In principle, $(A_p,B_p)$ could be uncontrollable even if $(A,B)$ is controllable, see \cite{Ogata1995}. 

The contracts ensure the planner that the lower-layer controller can bound the tracking error by 
\begin{align}
x\big((k_p+1)\cdot M\big)-x_p(k_p+1)\in\mathbb{Z}_i,
\label{eq:contract}
\end{align}
for $i=1,\ldots,N_w$, 
where the sets $\mathbb{Z}_i$ are convex, compact polytopes 
 and $N_w$ is the number of the operating mode. 

The lower-layer controller can guarantee the constraints, if  $x\in\mathbb{X}_i$ and $u\in\mathbb{U}_i$, i.e., that the state and the input are  inside the operation region $i$, see Assumption \ref{as:2}. Additionally, constraints due to the obstacle avoidance requirements and the different sampling times need to be satisfied, which are introduced in s second step.
We assume that the contracts - the operating regions/modes - are designed offline and known by both control layers. They depend on the  design of 
 of the low-layer controller and the operation mode $i$, i.e., the (partly) selectable uncertainty bound $\mathbb{W}_i$ on the disturbance $w$.
Summarizing: Utilizing the idea of contracts - different operating modes - enables  the planner to utilize and take the capability of the low-layer control loop into account for the computation of the reference. 
Consequently,  the planner  calculates and sends  to the low-layer controller a reference  and selects via 
the choice of $\mathbb{Z}_i$ the  maximum allowed tracking discrepancy. 
In other words, the reference planner can  switch between different operation modes, in order to 
improve the performance, as illustrated in Fig. \ref{fig:pf2new}.

\subsection{Upper-layer: Moving Horizon Reference Planning}\label{sec.plan}
 The reference planer should guarantee constraint satisfaction, including obstacle avoidance, despite  the presence of disturbances and uncertainties in combination with a suitably designed lower-layer control loop, compare Fig.~\ref{fig:MPCplanningcontrol}. 
 
 The key idea is to  incorporate the concept of contracts - different operating modes - into mathematical programming based moving horizon planning schemes\cite{trodden2008multi,ibrahim2019hierarchical,ibrahim2020learning,pinto2017risk,richards2006robust}. 
To do so, the reference planning problem $\mathcal{P}\big(x(k_pM)\big)$ is formulated on a moving horizon as a optimization problem:
\begin{subequations}
	\begin{align}
	& & \min_{\lbrace x_p \rbrace , \lbrace u_p \rbrace ,i}J_p\big(\lbrace x_p \rbrace , \lbrace u_p \rbrace \big)& , \label{eq:pp_jj}\\	
	 \text{s.t. }\nonumber \\
	&&i&\in\lbrace 1,\ldots,N_w \rbrace\\
	& & x( k_pM)-x_p(k_p|k_p) &\in \mathbb{Z}_i , \label{eq:pp_XO} \\
	\forall j &\in \{0,\ldots,N-1\}: & x_p(k_p+j+1|k_p) &= A_p x_p(k_p+j|k_p)+ B_p u_p(k_p+j|k_p), \label{eq:pp_x} \\
\forall j &\in \{0,\ldots,N-1\}: & x_p(k_p+j|k_p) &\in \mathbb{X}_i\ominus  \mathbb{Z}_i, \label{eq:pp_X} \\
\forall j &\in \{0,\ldots,N-1\}: & u_p(k_p+j|k_p) &\in \mathbb{U}_i\ominus  {K}\mathbb{Z}_i, \label{eq:pp_U} \\
\forall j &\in \{0,\ldots,N-1\}: & Cx_p(k_p+j|k_p) &\notin \mathbb{O}\oplus  (-C) \mathbb{Z}_i  \label{eq:pp_Y},\\
\forall j &\in \{0,\ldots,N-1\}: &	\big(x_p(k_p+j|k_p),u_p(k_p+j|k_p)\big) &\in \mathbb{I}_i, \label{eq:isc} \\
	& & x_p(k_p+N|k_p) &\in \mathbb{X}_i^f   . \label{eq:pp_Xf} 
	\end{align}
	\label{eq:milp}
\end{subequations}
\hide{\begin{subequations}
	\begin{align*}
	\min_{\lbrace x_p \rbrace , \lbrace u_p \rbrace ,i}J_p\big(\lbrace x_p \rbrace , \lbrace u_p \rbrace \big) ,& \label{eq:pp_jj}\\	
	\text{s.t. }  &   \nonumber \noindent \\	
	i&\in{1,\ldots,N_w}\\
	\\
		x( k_pM)-x_p(k_p|k_p) &\in \mathbb{Z}_i , \label{eq:pp_XO} \\
  x_p(k_p+j+1|k_p) &= A_p x_p(k_p+j|k_p)+ B_p u_p(k_p+j|k_p),\,  j= 1, \ldots, N-1 , \label{eq:pp_x} \\
 x_p(k_p+j|k_p) &\in \mathbb{X}_i\ominus  \mathbb{Z}_i,\,  j= 1, \ldots, N-1 ,  \label{eq:pp_X} \\
 u_p(k_p+j|k_p) &\in \mathbb{U}_i\ominus  {K}\mathbb{Z}_i,\,  j= 1, \ldots, N-1 ,  \label{eq:pp_U} \\
 Cx_p(k_p+j|k_p) &\notin \mathbb{O}\oplus  (-C) \mathbb{Z}_i,\,  j= 1, \ldots, N-1 ,  \label{eq:pp_Y}\\
	\big(x_p(k_p+j|k_p),u_p(k_p+j|k_p)\big) &\in \mathbb{I}_i,\,  j= 1, \ldots, N-1 ,  \label{eq:isc} \\
	x_p(k_p+N|k_p) &\in \mathbb{X}_i^f   . \label{eq:pp_Xf} 
	\end{align*}
	\label{eq:milp}
\end{subequations}}
Here  $N$ denotes the planning  horizon, $i$  the operation mode, and  $(k_p+j|k_p)$ corresponds to the prediction of a value at time $k_p + j$ made at time $k_p$. 
The constraint~\eqref{eq:pp_XO} represent an initial constraint at the begin of the planning time $k_p$. 
 Eq.~\eqref{eq:pp_x} represent the vehicle dynamics used by the planning layer. 
While the constraints~(\ref{eq:pp_X},\ref{eq:pp_U}) are the tightened  state  and input constraints,     where $K$ is  a control gain, which is discussed later in detail.
For the 
output constraints \eqref{eq:pp_Y}  the constraint tightening corresponds to an  enlargement of the obstacles, see Fig.~\ref{fig:pf2new}. 

To allow for an efficient reformulation of \eqref{eq:milp} as an MILP, c.f. Section~\ref{Scheduler}, we consider the following cost function
\begin{align*}
J_{\text{p}}\big(\lbrace x_p \rbrace , \lbrace u_p \rbrace \big) =  \|x_{\text{\tiny goal}} - x_p(k_p+N)  \|_\infty+ \sum_{j=k_p}^{k_p+N-1}\alpha_x\| {x}_p(j)\|_\infty + \alpha_u\|{u}_p(j)\|_\infty.
\end{align*} 
The stage cost penalizing the state $x_p$ and control input $u_p$ with different weights $(\alpha_x \geq 0,\alpha_u \geq 0)$, respectively. The terminal cost penalizes the vehicle's distance at the end of the planning horizon to the goal $x_{\text{goal}}$.
 
The inter-sample constraints \eqref{eq:isc} and the terminal constraint \eqref{eq:pp_Xf}  depend on the operation mode $i$ and are non-convex.
We make the following assumption with respect to the inter-sample constraints \eqref{eq:isc}  
\begin{assumption}(Inter-sample constraints) \label{ass:isc}
	The inter-sample constraints \eqref{eq:isc} determined by the sets $ \mathbb{I}_i$  are such that, if $(x_p ,u_p)\in \mathbb{I}_i$, then  for $  \ell=1,\ldots,M-1$ it holds
	\begin{subequations}\label{def_isc}    \begin{align}
		A^\ell x_p+\sum\limits_{m=0}^{\ell-1}A^mBu_p  &\in \mathbb{X}_i\ominus\mathbb{Z}_i,\\
		C(A^\ell x_p+\sum\limits_{m=0}^{\ell-1}A^mBu_p)&  \notin \mathbb{O} \oplus (-C)\mathbb{Z}_i.
		\end{align}\end{subequations}
\end{assumption} 
This assumption guarantees that the lower-layer/tracking control loop, operating at the faster time scale, can always satisfy the constraints. A straightforward choice is to choose $ \mathbb{I}_i$ directly as \eqref{def_isc}, which can lead to a large number of constraints and thus might increase the computational effort.
 Note that depending on the actual dynamics \eqref{eq:xp} certain constraints      in \eqref{eq:isc}  might be redundant and thus also
 the overall optimization \eqref{eq:milp} and can be removed without changing the solution of the optimization problem, e.g. using physical insight into the system dynamics or with the procedure presented in~\cite{Kerrigan2001}.

For the terminal sets  $\mathbb{X}^f_i$ we  assume  that they are positive invariant sets for \eqref{eq:xp} satisfying all constraints:
\begin{assumption}(Terminal sets and terminal control laws)\label{ass:terminal}~
	The  terminal control laws $\kappa^f_i(x_p)$ and the terminal sets $ \mathbb{X}^f_i$ are such  that 	$x_p \in \mathbb{X}^f_i$ implies:
	\begin{subequations}\begin{align}
		A_px_p+B_p\kappa^f_i(x_p)&\in\mathbb{X}_i^f, \\
		x_p&\in\mathbb{X}_i\ominus\mathbb{Z}_i, \\
		\kappa^f_i(x_p)&\in\mathbb{U}_i\ominus K \mathbb{Z}_i,\\
		\big(x_p ,\kappa^f_i(x_p)\big)&\in \mathbb{I}_i, \label{terminal_inter}\\
		Cx_p  &\notin \mathbb{O} \oplus (- C)\mathbb{Z}_i.\label{terminal_obstacle}
		\end{align} \end{subequations}
\end{assumption}
Clearly,  the terminal sets  $ \mathbb{X}^f_i$ are   non-convex due to the presence of the inter-sample constraints~\eqref{terminal_inter} and the obstacle avoidance condition~\eqref{terminal_obstacle}. 
A possible  choice   are  admissible, nominal steady states  $x_p=A_px_p+B_p\kappa^f_i(x_p)$ for the terminal control sets and the corresponding inputs as terminal control laws. 
For autonomous vehicles, these are basically all points where the vehicle can stop its motion safely. These points can also be determined for systems with complex dynamics such as unmanned aerial vehicles. Note that the terminal control laws $\kappa^f_i(x_p)$  are fictitious and never implemented.

The upper-layer planning algorithm solves the optimization problem $\mathcal{P}\big(x(k_pM)\big)$
\eqref{eq:milp}. Based on the optimal solution it sends to the  lower-layer controller  the chosen operation mode $i^\star$ and an inter-sampled reference 
\begin{align}\label{input_lower}
x_{\text{ref}}(k_pM+j)= A^jx_p^\star(k_p\vert k_p)+\sum\limits_{m=0}^{j} A^{m-1}Bu_p^\star(k_p\vert k_p),\,\, j=0,\ldots,M.
\end{align} 
Clearly, $\big(x_p^\star(k_p\vert k_p),u_p^\star(k_p\vert k_p)\big)\in\mathbb{I}_i$ implies that the reference  satisfies
\begin{align} 
C x_{\text{ref}}(k_pM+j)\notin \mathbb{O} \oplus(- C)\mathbb{Z}_i ,\,\,  j=0,\ldots,M, \label{eq:xref_int_cond}
\end{align} 
which means that the planned reference $x_{\text{ref}}$  robustly avoids obstacles, it satisfies the  condition \eqref{eq:all_obstacles}.
Using the idea of contracts between both layers, we can guarantee the following.
\begin{proposition}\label{robustprop}
	(Recursive feasibility of the upper layer planning) 
	Let   Assumptions \ref{as:2}-\ref{ass:terminal} hold
	and assume that	the lower-layer controller guarantees the contract  \eqref{eq:contract}   for the reference \eqref{input_lower}.
	If the planning optimization problem $\mathcal{P}\big(x(k_pM)\big)$ \eqref{eq:milp} is   feasible, then   also $\mathcal{P}\big(x((k_p+1)M)\big)$ is feasible.
\end{proposition}          
\begin{proof} 
The key idea for the proof is that the upper-layer planning MPC (in combination with the contracts) corresponds to a tube-based MPC using robust invariant sets, compare \cite{grune2017nonlinear,rawlings2017book,mayne2005robust}.

We denote the optimal solution of $\mathcal{P}\big(x(k_pM)\big)$   by   $x_p^\star(k_p\vert k_p)$, \ldots,   $u_p^\star(k_p\vert k_p)$,   \ldots,   $i^\star$. 
	Let us consider the following  guess as solution for the optimization problem $\mathcal{P}\big(x((k_p+1)M)\big)$ 
	\begin{align*}
	i&=i^\star, \\
	\forall j\in\{1,\ldots,N\}:x_p(k_p+j \vert k_p+1)&=x^\star_p(k_p+j\vert k_p), \\
	x_p(k_p+N+1 \vert k_p+1)&=A_px^\star_p(k_p+N\vert k_p)+B_p\kappa^f_i\big(x^\star_p(k_p+N\vert k_p)\big), \\
	\forall m\in\{1,\ldots,N-1\}:u_p(k_p+m \vert k_p+1)&=u^\star_p(k_p+m\vert k_p),\\
	u_p(k_p+N \vert k_p+1)&=\kappa^f_i\big(x^\star_p(k_p+N\vert k_p)\big).
	\end{align*}  
	Note that this guess is based on the previous solution, 
	the selected operation mode $i^\star$ and the terminal control law $\kappa_i^f$.
		We need to verify that this guess is feasible (but it might be possibly suboptimal) for the optimization problem $\mathcal{P}\big(x((k_p+1)M)\big)$.
	Using the contract \eqref{eq:contract}, i.e. the guarantee on the lower-layer control loop,   we have that $x\big((k_p+1)M\big)-x_p(k_p+1\vert k_p+1)\in\mathbb{Z}_i$, i.e. 	\eqref{eq:pp_XO} holds for $\mathcal{P}\big(x((k_p+1)M)\big)$.  
	The above initial guess 
	satisfies the constraints  \eqref{eq:pp_x}, \eqref{eq:pp_X}-\eqref{eq:isc}      for $1\leq j<N-1$
	for the 	optimization problem $\mathcal{P}\big(x(k_pM)\big)$, thus also the similar constraints for     $\mathcal{P}\big(x((k_p+1)M)\big)$.
	Finally, using the conditions on the terminal sets and terminal control laws in Assumption~\ref{ass:terminal} imply that  also the remaining constraints of $\mathcal{P}\big(x((k_p+1)M)\big)$, i.e.
	  \eqref{eq:pp_x}-\eqref{eq:isc}    with $j=N-1$ 	and \eqref{eq:pp_Xf} are satisfied.
\end{proof} 

\begin{remark}(Planning without feedback)
\emph{
In  \eqref{eq:milp} the initial constraint \eqref{eq:pp_XO}  provides  a feedback between the planning state $x_p$ and the real state $x$ for all $k_p$. One can only enforce the  constraint \eqref{eq:pp_XO}    at the initial time ($k_p=0$)  and replace it by the simpler 
 equality constraint $x^\star_p(k_p+1\vert k_p)=x_p(k_p+1\vert k_p+1)$ for $k_p>0$. 
		This  would remove the feedback from the plant to the upper-layer planner. This has the advantage  that it  avoids the need to wait for plant feedback for the planning and thus could enable computationally less restrictive planning, but it would also decrease the control   performance.} \label{rem:independent}
\end{remark}  

The optimization problem $\mathcal{P}$~\eqref{eq:milp} is non-convex, but it can be reformulated as an MILP, which can be efficiently solved, see  Section \ref{Scheduler}.


\subsection{Lower-layer: Cyclic Horizon Robust Model Predictive Tracking Control}\label{sec:PF}
The lower-layer controller tracks the generated reference based on the (faster) dynamics of the real system and needs to guarantee the contracts ~\eqref{eq:contract}. Also at the lower-layer we use the concept  of robust, tube-based MPC~\cite{grune2017nonlinear,rawlings2017book,chisci2001systems,mayne2005robust}, but we rely on growing tubes~\cite{chisci2001systems} instead of tubes based on robust positive invariance~\cite{mayne2005robust} as in the upper-layer.
The proposed tube-based MPC of the lower-layer is based on  a nominal prediction dynamics (nominal state $z$, nominal input $v$), which starts from the   real state $x(k)$ at the current time $k$:
\begin{subequations}\label{LL_c_1}
	\begin{align}
	z(k+j+1\vert k)&=Az(k+j\vert k)+Bv(k+j\vert k),\\
	z(k\vert k)&=x(k).
	\end{align} \end{subequations}
The  effect of   disturbances $w(k+j)$ onto the closed loop    is taken into account using a fictitious, auxiliary control law  of the form
\begin{align}\label{aux_control_law}
u(k+j\vert k)=v(k+j\vert k)+K\big(x(k+j)-z(k+j\vert k)\big). 
\end{align}
The control gain $K$ in this affine feedback is chosen such that $A+BK$ is Schur stable.
The  auxiliary control law is utilized  to determine sets to bound  the difference $e(k+j\vert k)=x(k+j)-z(k+j\vert k)$
between the real system state $x$ and the predictions $z$ made using~\eqref{LL_c_1}.  In detail, for the $i$-th operation mode the error bounds satisfy  $e(k+j\vert k)\in\mathbb{E}_i
(j)$ where 
\begin{align}
\mathbb{E}_i(j+1)=(A+BK)  \mathbb{E}_i(j)\oplus \mathbb{W}_i, &&
\mathbb{E}_i(0)=\lbrace 0 \rbrace. \label{eq:defE}
\end{align}
Note that the size of the sets $\mathbb{E}_i$ monotonically increases with $j$, i.e. $\mathbb{E}_i(j)\subseteq \mathbb{E}_{i}(j+1)$. However,      for any $j\geq 0$ we have that
$\mathbb{E}_i(j) \subseteq \mathbb{Z}_i$, where $\mathbb{Z}_i$ is the (minimum) robust positive invariant set, compare \cite{rawlings2017book}, satisfying:
\begin{align}
\mathbb{Z}_i\supseteq (A+BK) \mathbb{Z}_i \oplus \mathbb{W}_i .\label{eq:defZi}
\end{align} 
In the lower-layer MPC we predict until the next planning instant utilizing a cyclic horizon  $L_k=M-\operatorname{rem}(k,M)$, see \cite{kogel2013stability}. For the case that $k$ is a multiple of $M$, we have $L_k=M$. Otherwise, $L_k$ is smaller than $M$, but $k+L_k$ is a multiple of $M$. Consequently,  the horizon shrinks between two planning instants and is increased at the next planning instant again to length $M$.


The lower-layer MPC predicts and optimizes    nominal  state  and input sequences
\begin{align}
\mathbf{z}(k)=\lbrace z(k \vert k),\ldots, z(k+L_k \vert k)\rbrace,\quad\mathbf{v}( k)=\lbrace v(k \vert k),\ldots, v(k+L_k-1 \vert k)\rbrace, \end{align}
based on the nominal dynamics~\eqref{LL_c_1} and subject to satisfaction of the constraints
\begin{subequations}\label{LL_c_2}
	\begin{align}
	z(k+j\vert k) &\in \mathbb{X}_i \ominus \mathbb{E}_i(j ), \\  
	v(k+j\vert k) &\in \mathbb{U}_i \ominus K\mathbb{E}_i(j ), \\
C\left( z(k+j\vert k)-x_{\text{ref}}(k+j)\right)     &\in    C\left( \mathbb{Z}_i \ominus \mathbb{E}_i(j)\right)  ,\label{LL_c_2_c} \\
	z(k+L_k\vert k)-x_{\text{ref}}(k+L_k) &\in    \mathbb{Z}_i \ominus \mathbb{E}_i(L_k). \label{LL_c_2_d}
	\end{align}
\end{subequations}
Note that these constraints include   the convex state and input constraints~\eqref{state_input_constraints}.
In contrast, the non-convex obstacle avoidance constraints are taken into account  using the concept of contracts. Basically, the lower-layer controller needs to  enforce the guaranteed accuracy with respect to  the output (condition~\eqref{LL_c_2_c}) or even the full state at  the end of the prediction (condition~\eqref{LL_c_2_d}). Note that, the constraints~\eqref{LL_c_1},~\eqref{LL_c_2} are convex. 

The lower-layer MPC 
penalizes the deviation error from the reference $x_{\text{ref}}$ and 
utilizes the convex cost function 
\begin{align}
J_\text{t}\big( \mathbf{z}( k),\mathbf{v}( k) \big)  =& \sum_{j=k}^{k+L_k-1}   \|x_{\text{ref}}(j) -  z(j\vert k)  \|_Q^2  +\| v(j\vert k) \|_R^2 +\|x_{\text{ref}}(k+L_k) - z(k+L_k\vert k)  \|_P^2 ,
\end{align}
where  the matrices  
$Q \in \mathbb{R}^{n\times n}$, $P \in \mathbb{R}^{n\times n}$ and $R \in \mathbb{R}^{m\times m}$
are positive definite and represent the  weighting for the inter-sample states, the final state and the inputs, respectively. 

The applied control  input $u(k)=v^\star(k\vert k)$  is given by solving  the  optimization problem  $\mathcal{L}\big(x(k),\lbrace x_{\text{ref}}\rbrace,i,k\big)$
\begin{align}\min\limits_{\mathbf{z}( k),\mathbf{v}( k)}J_\text{t}\big(
\mathbf{z}( k),\mathbf{v}( k)  \big) \mbox{ s.t.~\eqref{LL_c_1},~\eqref{LL_c_2}.} \label{OCP_LL}
\end{align}
This  optimization problem depends on the current state available at the lower-layer as well as   the  reference and the operation mode determined by the upper-layer planner. The resulting optimization problem is a convex quadratic program (QP) and has, in addition, a special structure, which allows its efficient solution, even on computationally limited hardware, see e.g. \cite{wang2009fast,zometa2013muao}.

For the lower-layer, we can derive  the following   properties assuming that the upper-layer reference is chosen suitably.
\begin{proposition}(Constraint satisfaction and obstacle avoidance) 
	Let Assumptions \ref{as:2} and \ref{ass:isc} hold and assume that the reference $x_{\text{ref}}$ is given by~\eqref{input_lower} and satisfies~\eqref{eq:xref_int_cond}.
	If the lower-layer MPC problem
	$\mathcal{L}\big(x(k),\lbrace x_{\text{ref}}\rbrace,i,k\big)$~\eqref{OCP_LL} is feasible, then the constraints~\eqref{state_input_constraints} are satisfied: $x(k)\in\mathbb{X}_i\subseteq \mathbb{X}$ and $u(k)\in\mathbb{U}_i\subseteq \mathbb{U}$ and the obstacles are avoided, i.e.~\eqref{obstacle_constraints} holds.
\end{proposition}

\begin{proof}
Combining ~\eqref{LL_c_1},~\eqref{aux_control_law} and~\eqref{LL_c_2} we have that 
	$z^\star(k\vert k)=x(k)$,  $v^\star(k\vert k)=u(k)$ and  $Cz^\star(k\vert k)=y(k)$. Together with $\mathbb{E}_i(0)=\lbrace 0 \rbrace$ and Assumption \ref{as:2} this implies that the state and input constraints~\eqref{state_input_constraints} are satisfied. 
	\eqref{LL_c_2_c} together with 	\eqref{eq:defE} yield that  $y(k)-Cx_{\text{ref}}(k) \in C\mathbb{Z}_i$. Combined with~\eqref{eq:pp_Y}
\eqref{eq:xref_int_cond} this implies that  the avoidance constraint $y_k\notin  \mathbb{O}$ ~\eqref{obstacle_constraints} holds. 
\end{proof} 
Note that if the upper-layer planner provides the reference  $x_{\text{ref}}$ ~\eqref{input_lower}, then also \eqref{eq:xref_int_cond} holds.

\begin{proposition}(Recursive feasibility of  the overall control scheme) 
	Let Assumptions \ref{as:2}-\ref{ass:terminal} hold. If the planning problem~\eqref{eq:milp} is feasible for $x(0)$, then the optimization problems ~\eqref{eq:milp},~\eqref{OCP_LL} remain feasible
	for the closed loop system consisting of the upper-layer moving horizon  planner~\eqref{eq:milp},~\eqref{input_lower}, the lower-layer controller~\eqref{OCP_LL} and the uncertain plant dynamics~\eqref{system}.
\end{proposition}

\begin{proof}
	The proof consists of three parts:  first it is shown that feasibility of the planning problem $\mathcal{P}\big(x(k_pM)\big)$~\eqref{eq:milp} implies feasibility of the low-layer optimization problem~\eqref{OCP_LL} $\mathcal{L}\big(x(k_pM),\lbrace x_{\text{ref}}\rbrace,i^\star,k_pM\big)$. Afterwards we verify that feasibility of
	$\mathcal{L}\big(x(k),\lbrace x_{\text{ref}}\rbrace,i,k\big)$ 
	\eqref{OCP_LL} implies: if $k+1$ is not a multiple of $M$, feasibility of  $\mathcal{L}\big(x(k+1),\lbrace x_{\text{ref}}\rbrace,i,k+1\big)$~\eqref{OCP_LL} (part 2) and otherwise feasibility of  the planning problem $\mathcal{P}\big(x(k+1)\big)$ \eqref{eq:milp} (part 3).
	
	1) If   $\mathcal{P}\big(x(k_pM)$~\eqref{eq:milp} is feasible, then  the following (suboptimal) input trajectory and state sequence   for the lower-layer problem 	 $\mathcal{L}\big(x(k_pM),\lbrace x_{\text{ref}}\rbrace,i^\star,k_pM\big)$
	\begin{align*}
	z(k\vert k)=&x(k) \\
	v(k+j\vert k)=& u_p^\star(k_p\vert k_p)+K(z(k+j\vert k)-x_\text{ref}(k+j\vert k)),\,\,j=0,\ldots,M-1, \\
	z(k+j+1\vert k)=&Az(k+j\vert k)+B  v(k+j\vert k),\,\,j=0,\ldots,M-1, 
	\end{align*} 
  satisfies all constraints of \eqref{OCP_LL} due to the inter-sample constraints $\mathbb{I}_i$ and the  consistent constraint tightening utilized at both layers, i.e., the definition of $\mathbb{E}_i$ and $\mathbb{Z}_i$, see
  \eqref{eq:defE} and \eqref{eq:defZi} and
  that $\mathbb{E}_i\subseteq \mathbb{Z}_i$.
	
	2) In the case that $k+1$ is not a multiple of $M$, i.e. no planning takes place,  the horizon $L_k$ shrinks. Due to the design of the set $\mathbb{E}_i(j)$~\eqref{eq:defE}   a feasible nominal state trajectory $\mathbf{z}(k+1)$ and a nominal input trajectory $\mathbf{v}(k+1)$
	for 	$\mathcal{L}\big(x(k+1),\lbrace x_{\text{ref}}\rbrace,i,k+1\big)$, satisfying 
	\eqref{LL_c_1} and \eqref{LL_c_2}, can be obtained from the solution of 
	$\mathcal{L}\big(x(k),\lbrace x_{\text{ref}}\rbrace,i,k\big)$:
	\begin{align*}
	z(k+j\vert k+1)=&z^\star(k+j\vert k)+(A+BK)^{j-1}w(k),\\
	v(k+j\vert k+1)=&v^\star(k+j\vert k)+K(A+BK)^{j-1}w(k).
	\end{align*}
	
	3) Finally, if $k+1$ is a multiple of $M$, i.e. the planning problem is solved at $k+1$, then feasibility
	of  	$\mathcal{L}\big(x(k),\lbrace x_{\text{ref}}\rbrace,i,k\big)$, in particular \eqref{LL_c_2_d} 
	implies 
	that  $x(k+1)-x_{\text{ref}}(k+1)\in\mathbb{Z}_{i^\star}$. This   together with Proposition \ref{robustprop} yields that the planning problem \eqref{eq:milp} is feasible at $k+1$. 
\end{proof}
\begin{remark}(Adaption of sets)~ 
\emph{
	We assume that  state and input  constraints $\{\mathbb{X}_i,\mathbb{U}_i\}$ and the tubes/contracts $\{\mathbb{E}_i,\mathbb{Z}_i\}$ are determined offline. The proposed approach can in principle be extended to allow an adaption of these sets. This could be useful for example to consider   the influence of varying  weather conditions. We do not consider such an extension in this work.}
\end{remark} 

\begin{remark}(Relaxing condition \eqref{LL_c_2_c})
\emph{
In the optimization problem \eqref{OCP_LL} the difference between the real/predicted output $y_k$/$z_k$ 
and its reference $Cx_{\text{ref}}(k)$     is restricted to the set
$C\mathbb{Z}_i \ominus C \mathbb{E}_i$, compare 
\eqref{LL_c_2_c}. This restriction is used to enable  guarantees on  the obstacle avoidance constraints \eqref{eq:all_obstacles}.
However, if the vehicle position/output (or /certain directions of it) is at time instant $k$ far way from an obstacle, then this restriction can be conservative. In principle, it is possible to relax these  constraints by generating online based on the solution of the upper-layer  sets $\mathbb{Y}_{k+j}$ resulting in less conservative output constraints:
 $Cz(k+j\vert k)\in \mathbb{Y}_{k+j}\ominus C\mathbb{E}_i$.
}\end{remark}

\begin{remark}(More general tube scheme)~
\emph{
We use a  basic tube scheme with a single gain $K$ and focused on linear dynamics. One could use a more general scheme with multiple gains,   see~\cite{Kogel2020robust},
or a more complex tube control law, see e.g. \cite{rakovic2011fully,Rakovic2012invention}. Also, an extension to nonlinear lower-layer dynamics is in principle possible using for example \cite{Koegel2015acc,villanueva2017computing}.}
\end{remark}

%
%
%

\subsection{MILP solution of the planning problem} \label{Scheduler}
In the following, we discuss how the non-convex  optimization problem $\mathcal{P}$~\eqref{eq:milp} can be reformulated using the big-M method \cite{ibrahim2019hierarchical,ibrahim2020learning} into an MILP. Note that $\mathcal{P}$~\eqref{eq:milp} is
 non-convex due to two reasons: firstly, due to the operation mode $i$,  and secondly   the non-convex obstacles avoidance constraints \eqref{obstacle_constraints} result in the non-convex constraints \eqref{eq:pp_U} - \eqref{eq:pp_Xf}.
  
\emph{Scheduling of operation modes}: In the proposed hierarchical scheme the operating modes of the vehicle given in the form of different constraint sets $\mathbb{X}_i$ and $\mathbb{U}_i$ directly influences 
the uncertainty bounds $w \in \mathbb{W}_i$, see Assumption~\ref{as:2}. 
 The lower-layer controller guarantees  constraint satisfaction and guarantees   bounds on the tracking error in form of a set   $\mathbb{Z}_i$, which depends on the    operation mode.  So, the sets appearing in the   initial constraint \eqref{eq:pp_XO}  and the tightened state/input constraints \eqref{eq:pp_X}, \eqref{eq:pp_U}  are of the form
\begin{align*}
\mathbb{Z}_i&\equiv \{x|F^z_i x \leq G^z_i\},\quad \forall i \in \{1,\ldots,N_w\},\\
\mathbb{X}_i\ominus \mathbb{Z}_i &\equiv \{x|F^x_i x \leq G^x_i\},\quad \forall i \in \{1,\ldots,N_w\},\\
\mathbb{U}_i\ominus K\mathbb{Z}_i &\equiv \{u|F^u_i u \leq G^u_i\},\quad \forall i \in \{1,\ldots,N_w\}.
\end{align*}
As result, this contract provides the planner an extra degree of freedom to reduce the planning conservatism by switching between different operating modes of the lower-layer controller.
\hide{
\color{magenta}
In this work, we use the big-M method to formulate the modes scheduling as  M-sided polygons:  
\begin{align*}
    F^x_i x \leq G^x_i +& M_{\text{big}}\big(1-d_i(k)\big),\quad \forall i \in \{1,\ldots,N_w\}, 
    \\F^u_i u \leq G^u_i +& M_{\text{big}}\big(1-d_i(k)\big),\quad \forall i \in \{1,\ldots,N_w\}, 
\end{align*} 
we use a large positive number $M_\text{big}$ to deactivate the mode $i$-th at time $k$ by relaxing its constraints using the binary decision variable $d_i(k)$.

Practically, the vehicle has to move at only one mode, so we exploit an extra constraint
$$\sum_{i = 1}^{N_w}d_i(k) = 1,$$
to insure one active mode at each planning sample.

\color{red}
Is there only one $d$? Why do we need the dependence on $k$? We also have the initial constraint \eqref{eq:pp_XO}, which depends on $i$. 
MK: Should we add that one could alternatively solve for each mode its problem in parallel?
\color{black} }

We use the so called big-M method to formulate the mode scheduling as:  
\begin{align*}
    F^z_i x-x_p \leq &G^z_i + M_{\text{big}} \big(1-d_i\big)\mathbf{1},\quad \forall i \in \{1,\ldots,N_w\}, \\
    F^x_i x \leq &G^x_i + M_{\text{big}} \big(1-d_i\big)\mathbf{1},\quad \forall i \in \{1,\ldots,N_w\},     \\
    F^u_i u \leq & G^u_i + M_{\text{big}} \big(1-d_i\big)\mathbf{1},\quad \forall i \in \{1,\ldots,N_w\},  \\
    \sum_{i = 1}^{N_w}d_i=&1.
\end{align*} 
Here we use a large positive number $M_\text{big}$ to deactivate the constraints of the  $i$-th mode  by relaxing its constraints using the binary decision variable $d_i$. The  last constraint guarantees that exactly one  mode is active in the planning. 

\emph{Obstacle avoidance constraints:}
\hide{
\color{magenta}
For each mode, the tracking error set $\mathbb{Z}_i$ is used to tight the state constraint $\mathbb{X}_i$ and to enlarge boundaries of the obstacle $\mathbb{O}$, compare \eqref{eq:milp}.
In this work, the non-convex avoidance constraints \eqref{obstacle_constraints}  are approximated by convex polygons using MILP with additional binary variables $b^i_s(k)$:\\
\begin{align*}
\forall s \in \{1,\ldots,S\}, \forall k \in \{1,\ldots,N\}&, i \in \{1,...,N_{w}\}:\\
\big(p_x(k) - p_x^o(k)\big)\cdot\cos\left(\dfrac{2\pi s}{S}\right)+&\big(p_y(k) - p_y^o(k)\big)\cdot\sin\left(\dfrac{2\pi s}{S}\right)\geq \delta^i_\text{safe} - M_\text{big}b^i_s(k).
\end{align*}
where we add a safety distance $\delta^i_\text{safe}$ between the vehicle position $\big(p_x(k),p_y(k)\big)$ and the obstacle $\big(p_x^o(k),p_y^o(k)\big)$ for the vehicle mode $i$.
Also, we impose an extra constraint 
\begin{align*}
\sum_{s=1}^{S} b^i_s(k) \leq S-d_i(k),
\end{align*}
to ensure that at least one active constraint for mode $i$, \ie\ $d_i(k) = 1$.
So the the planner can choose different low-layer
controller modes $i$ to adjust the obstacle boundary $\delta^i_\text{safe}$ by changing the vehicle velocity, see Fig. \ref{fig:tube_2d_3d}.

%
}
For each mode, the tracking error set $\mathbb{Z}_i$ is used to   enlarge boundaries of  the obstacles $\mathbb{O}$, compare \eqref{eq:milp}.
The enlarged,  non-convex avoidance constraints \eqref{obstacle_constraints} 
can be rewritten/over-approximated by 
\begin{align*}
 \forall j=0, \ldots,N-1,\,\forall, \ell=1\ldots,H  : 
 \exists a\in 1,\ldots, q_j\mbox{ s.t. }
 E_{\ell,a}Cx_p(k_p+j\vert k_p)\geq f_{\ell,a}^i.
\end{align*}
In this case one can enforce this constraint for the active mode $i$ by  using   additional binary variables $b^i_{\ell,a}(j)$ by requiring 
for $  j=0, \ldots,N-1$ and $  \ell=1\ldots,H$ 
that  
\begin{align*}
E_{\ell,a}Cx_p(k_p+j\vert k_p)&\geq f_{\ell,a}^i-M_{\text{big}} b^i_{\ell,a}(j),\\
\sum\limits_{a=1}^{q_j} b^i_{\ell,a}(j) &\leq q ^\ell-d_i.
\end{align*}
Here  we impose an extra constraint
to ensure that at least one active constraint for mode $i$, where $q ^\ell$ is number of faces of each obstacle. 
Loosely speaking  the planner can choose different low-layer
controller modes $i$ to adjust the obstacle boundary $f_{\ell,a}^i$ by changing the operation modes. This is illustrated in Fig. \ref{fig:label_two_modes}(b) considering that the operation modes corresponds to different velocities. 

\section{Unmanned Aerial Vehicle  Example}\label{sec:ex_res}
We consider a quadcopter that should fly from a starting to a goal point without hitting obstacles, c.f.  Figure \ref{fig:coordinates}.
\begin{figure}
	\begin{center}
		\includegraphics[width=0.75\columnwidth]{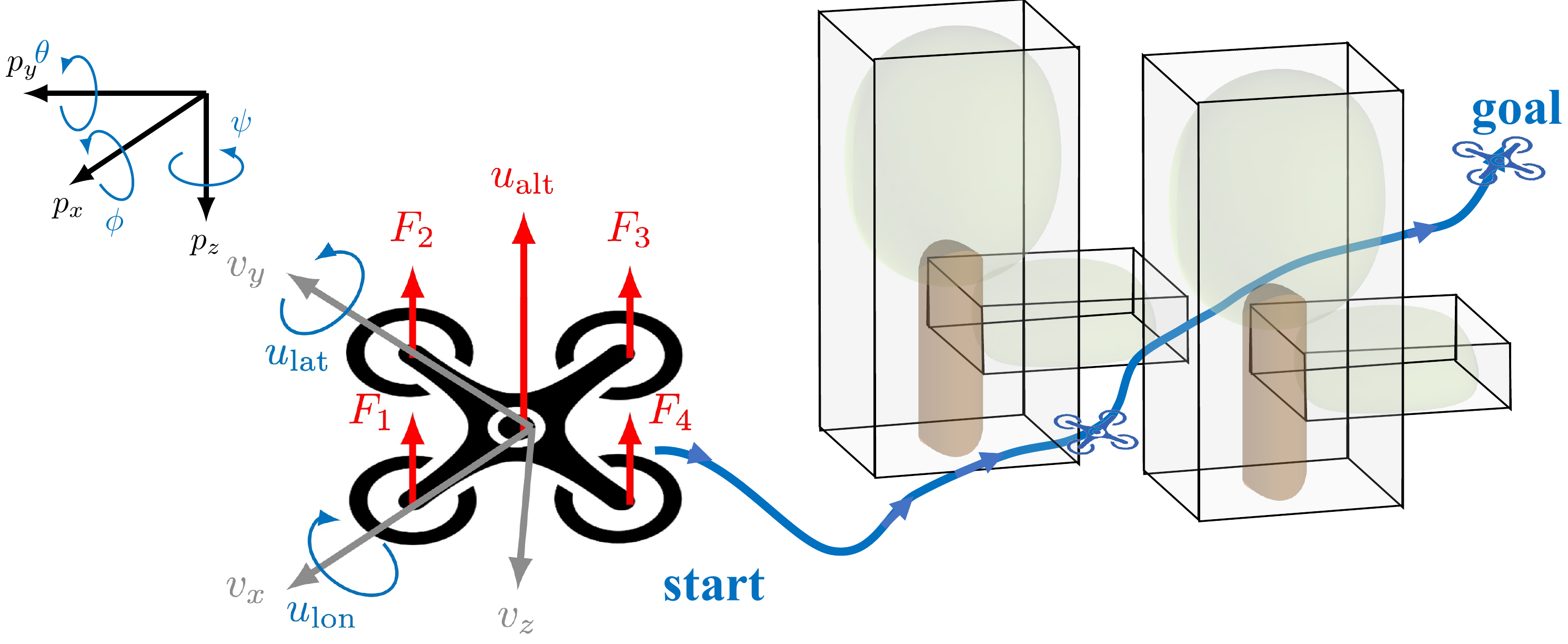}
		\caption{Illustration of the quadcopter state, forces, and moments described in the earth ($p_x$, $p_y$, $p_z$, $\theta$, $\varphi$, $\phi$) and body frame. The considered inputs are the lateral and longitudinal moments are provided by the four rotors resulting in the forces $F_i$.}  
		\label{fig:coordinates}
	\end{center}	
\end{figure}
A linearized model, as presented in \cite{alexis2016robust} based on a  nonlinear dynamic model is used, c.f. Appendix \ref{sec:ex}. The states of the quadcopter are roll and pitch angles $(\phi,\theta)$, roll and pitch rates $(w_x,w_y)$, 3D position $(p_x,p_y,p_z)$, and 3D velocities $(v_x,v_y,v_z)$.  The linearized model has three inputs: $u_{\text{lon}} = \theta_c$, $u_{\text{lat}} = \phi_c$ and $u_{\text{alt}} = T_c$.
The input  and states are constrained to:
\begin{align}
\left [\begin{smallmatrix}-1.5\, \text{m/s} \\ -\pi/4\, \text{rad} \\ -\pi/4\, \text{rad}	\end{smallmatrix} \right ]\leq
\left [\begin{smallmatrix} v_x, \, v_y, \, v_z\\\phi, \, \theta, \\ \phi_c, \, \theta_c  \end{smallmatrix} \right ]\leq
\left [\begin{smallmatrix}1.5\, \text{m/s} \\ \pi/4\, \text{rad} \\ \pi/4\, \text{rad}	\end{smallmatrix} \right ].
\end{align}

\hide{\subsection{Simulation Results}\label{sec:sim}}
The model is discretized with a  sampling time of $T_p = 0.5s$ for the planning problem and $T_t = 0.05s$, i.e. $M=10$, for the cyclic horizon MPC controller \eqref{eq:apa}. YALMIP~\cite{lofberg2004yalmip} is used to implement and formulate the planning and control problems, exploiting Gurobi~\cite{gurobi} for the solution of the optimization problems. The required sets are calculated via the MPT toolbox~\cite{mpt}. 

Figure \ref{fig:tube_2d_3d_nor} shows simulations for the case that only a single low layer control mode is used. 
\begin{figure}[htb]
	\centering
	\captionsetup[subfloat]{font=normalsize,labelfont={bf,sf}}
		\includegraphics[width=1\linewidth]{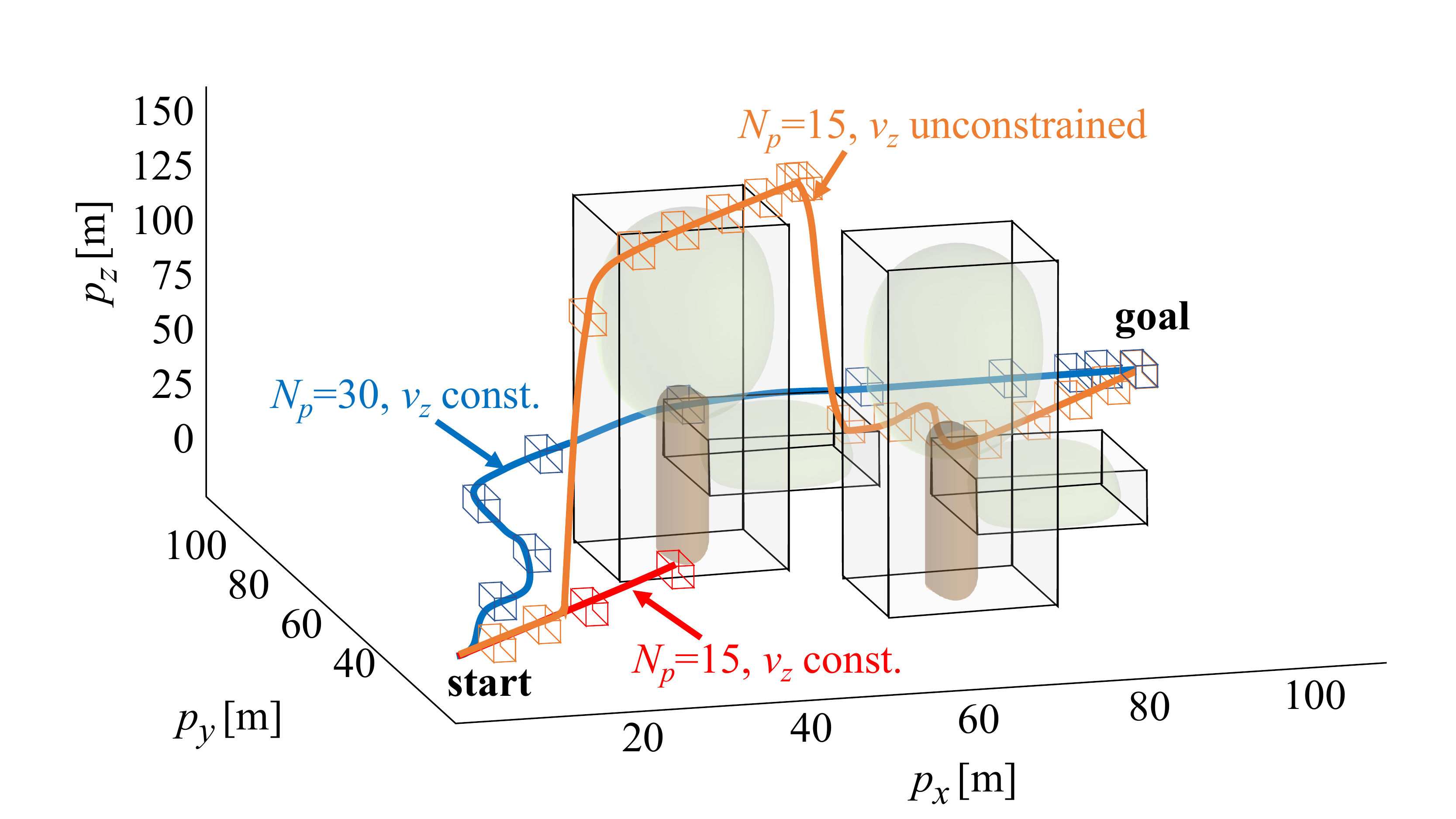}
	\caption{Disturbance free closed-loop simulation results for a single MPC controller mode. 
	}\label{fig:tube_2d_3d_nor} 
\end{figure}
As can be seen, for a small planning horizon of $N_p=15$ no path around the obstacle can be found. Only an increase of the planning horizon to $N_p=30$, or the removal of the vertical velocity constraint on $v_z$  allows the controlled vehicle to reach the goal. Note that in all cases no collisions occur. The (maximum) computation times to solve the planning problem are $\text{T}_\text{com} =$ 0.4 s for $N=15$, and $\text{T}_\text{com} =$ 60 s for  $N=30$, which is well above the desired re-planning time of 0.5 s.\footnote{The computation times are carried out on an Intel$^\circledR$ Core$^{TM}$ i7-8550U CPU which operates on 1.99GHz.}.

Figure \ref{fig:label_two_modes} shows simulation results using two operation modes for the low-layer MPC controller. They correspond to a fast operation mode, given by  $-1.5\, \text{m/s} \, \leq (v_x, \, v_y)\leq 1.5\, \text{m/s}$ and a slow operation mode given by $-1.0\, \text{m/s} \, \leq (v_x, \, v_y)\leq 1.0\, \text{m/s}$.  The fast operation mode corresponds to a large uncertainty set, while slow operation mode  leads to a smaller uncertainty bound $\mathbb{W}_2$ and thus smaller sets $\mathbb{E}_2$ and $\mathbb{Z}_2$.
\begin{figure}[htb]
\centering
\begin{minipage}[c]{.49\textwidth}
  \centering
     \begin{subfigure}[b]{\textwidth}
         \centering
         \includegraphics[width=\textwidth]{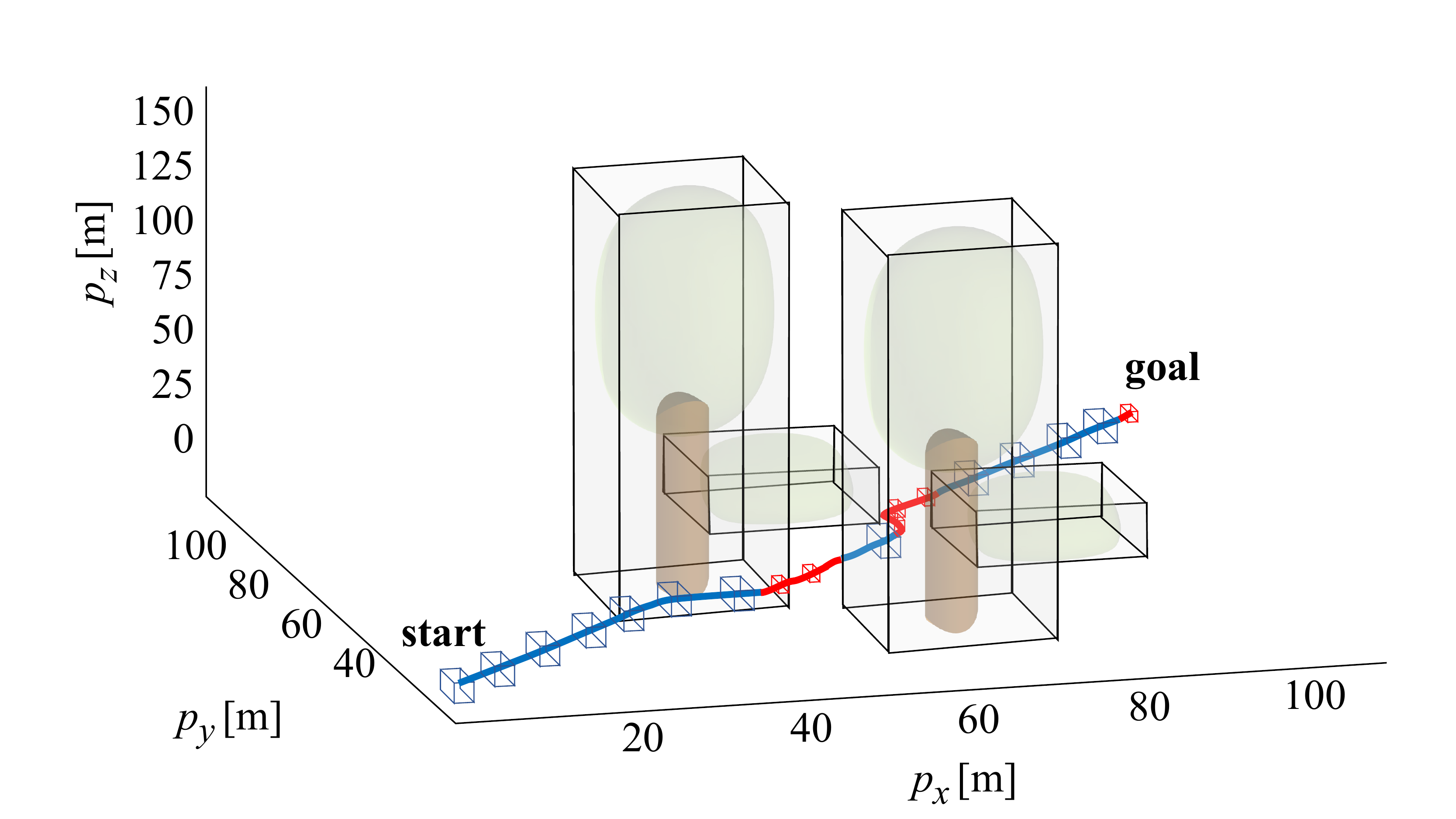}
      \end{subfigure}     
          \\
     \begin{subfigure}[b]{\textwidth}
         \centering
         \includegraphics[width=\textwidth]{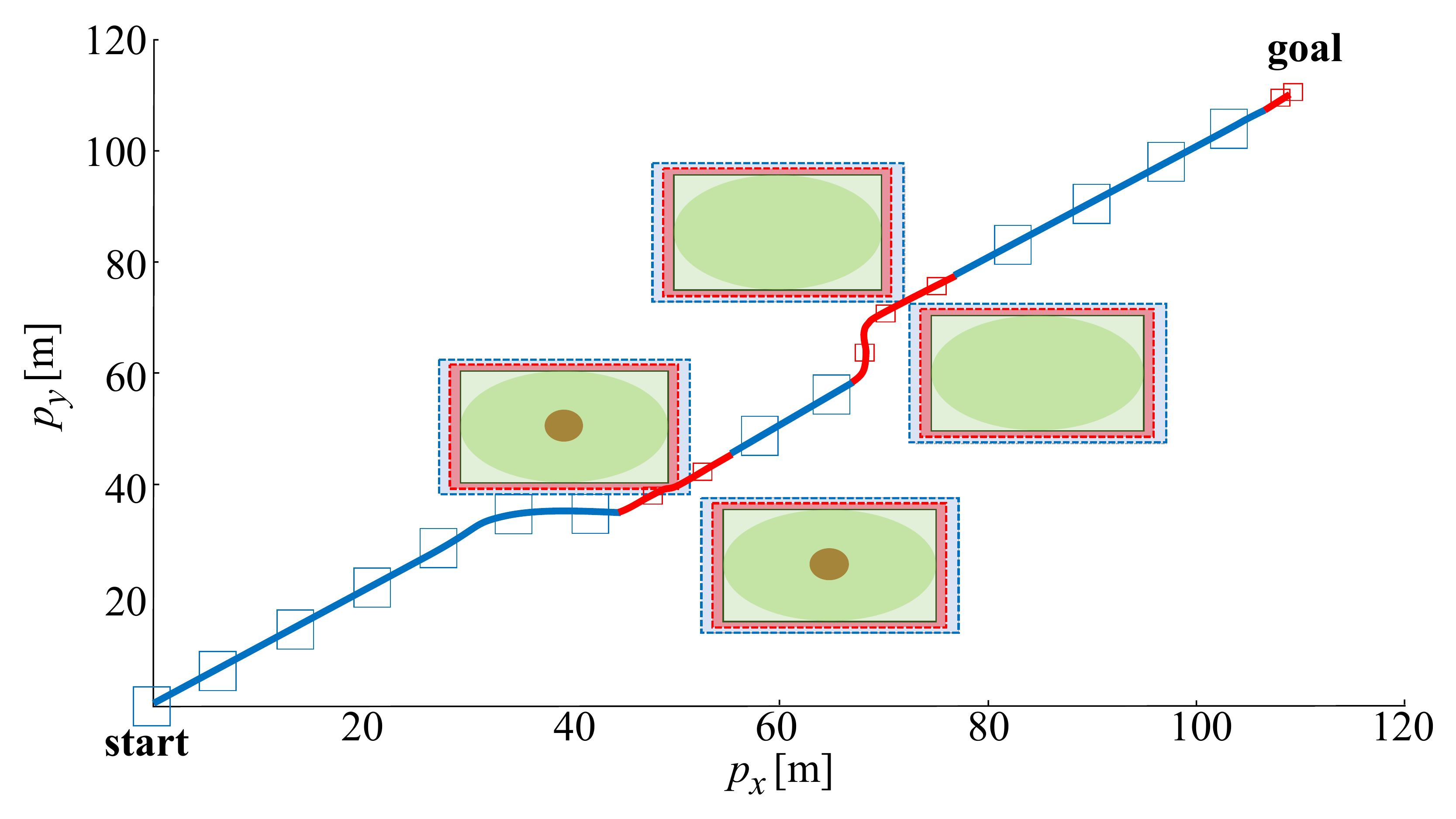}
        \captionsetup{font=normalsize,labelfont={bf,sf}}
  \caption{3D and top down view. } 
  \end{subfigure}
\end{minipage}%
\hfill
\begin{minipage}[c]{.49\textwidth}
  \centering
         \begin{subfigure}[b]{\textwidth}
            \includegraphics[width=\linewidth]{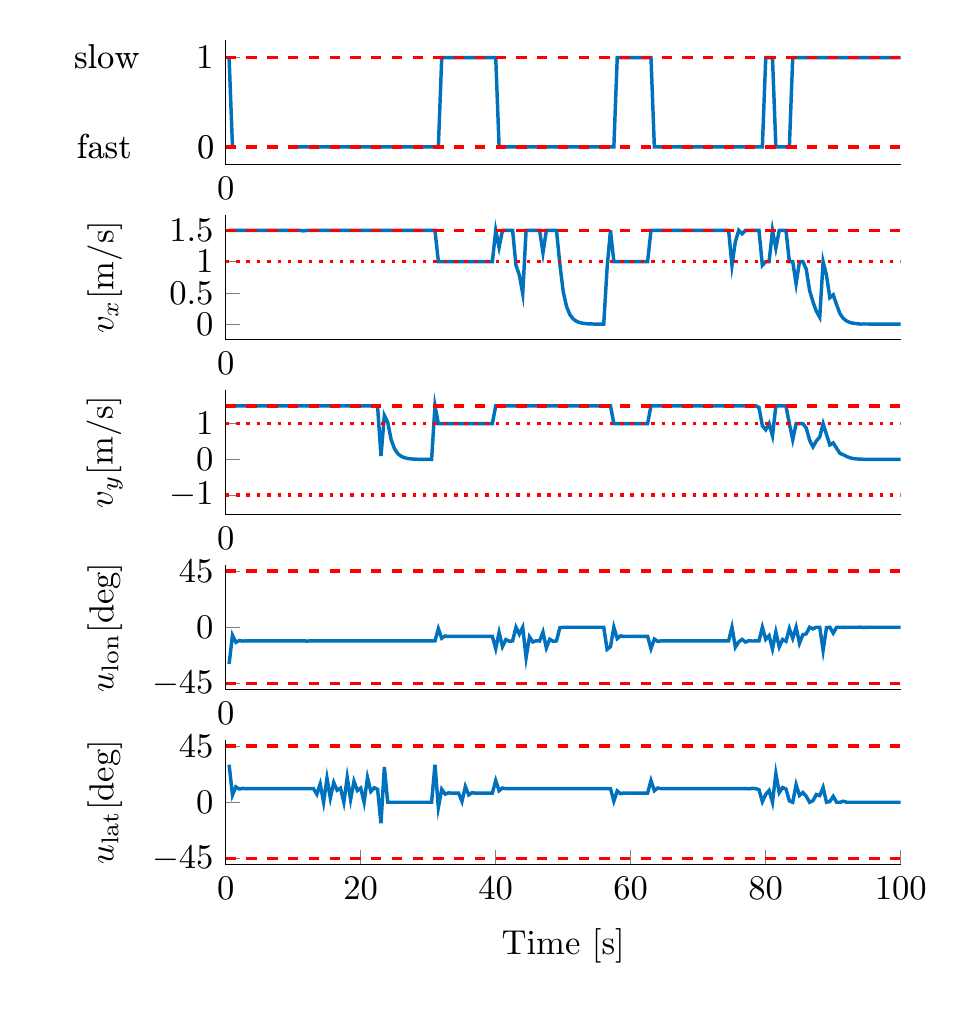}
        \captionsetup{font=normalsize,labelfont={bf,sf}}
      \caption{Selected operation mode, input, and speed profile}
      \end{subfigure}     
\end{minipage}
 \captionof{figure}{(a): Using two operating modes allows the quadcopter to reach the goal on a direct way. 
  (b): The reference planner switches between the two operation modes, resulting in different velocity constraints and uncertainty bounds.}\label{fig:label_two_modes}
\end{figure}

Figure \ref{fig:disturbances} shows simulation results for large wind disturbances, which are not considered in the controller explicitly. As can be seen, the hierarchical control strategy is able to achieve the goal, while avoiding the obstacles and satisfying the input constraints. 
\begin{figure}
\centering
\begin{minipage}[c]{.49\textwidth}
  \centering
     \begin{subfigure}[b]{\textwidth}
         \centering
         \includegraphics[width=\textwidth]{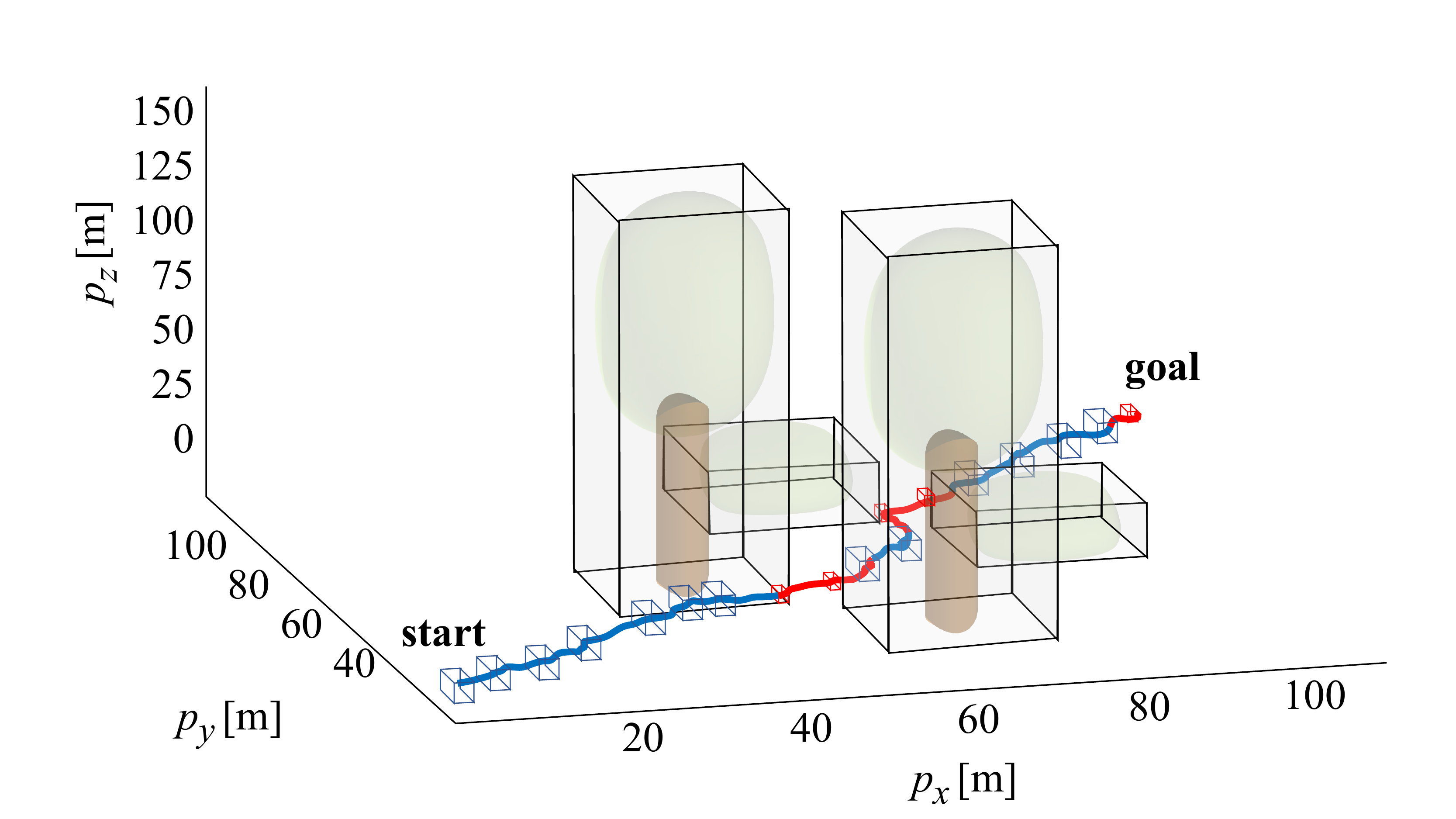}
         \captionsetup{font=normalsize,labelfont={bf,sf}}
     \end{subfigure}
    \begin{subfigure}[b]{\textwidth}
         \centering
         \includegraphics[width=\textwidth]{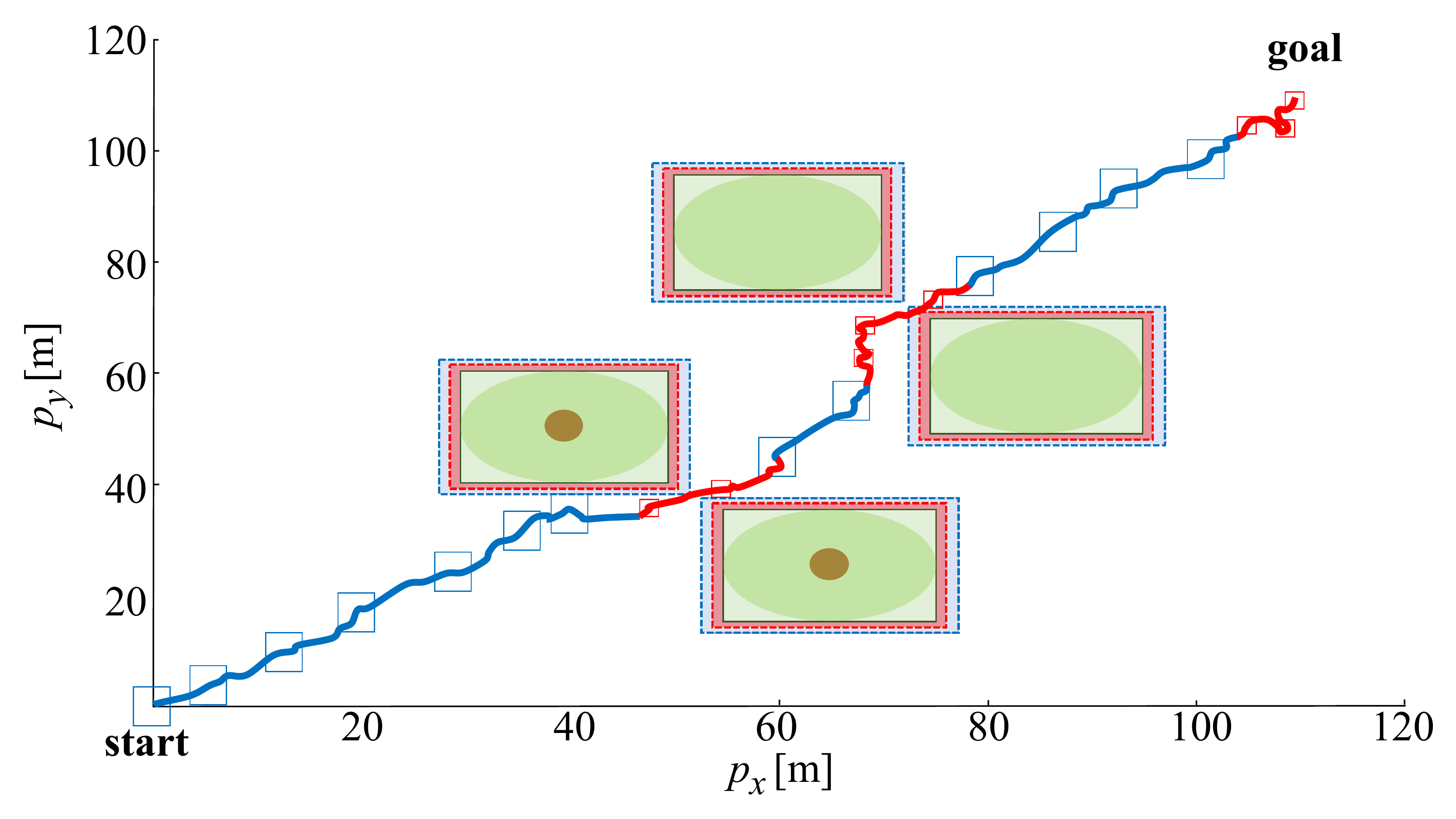}
         \captionsetup{font=normalsize,labelfont={bf,sf}}
         \caption{3D and top down view.}
     \end{subfigure}
\end{minipage}%
\hfill
\begin{minipage}[c]{.49\textwidth}
     \begin{subfigure}[b]{\textwidth}
  \centering
  \includegraphics[width=\linewidth]{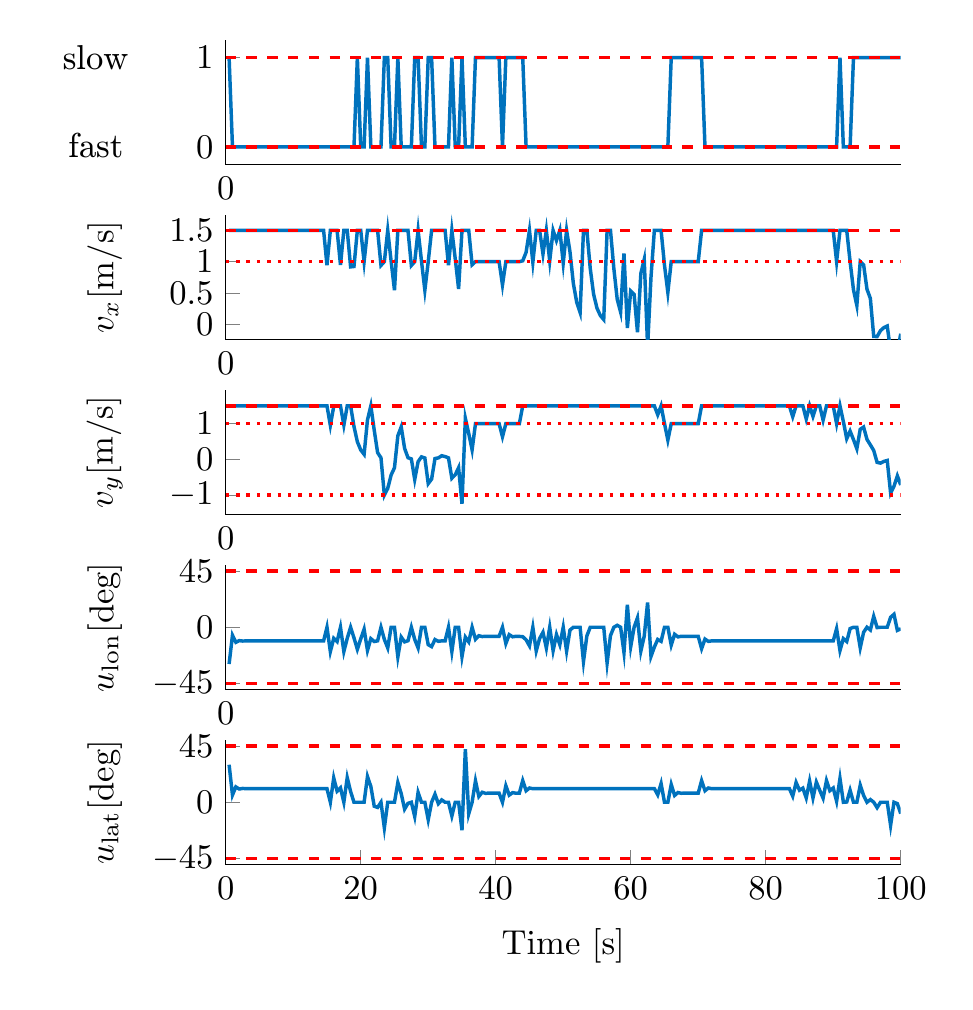}
        \captionsetup{font=normalsize,labelfont={bf,sf}}
  \caption{Selected operation mode, input, and speed profile}
\end{subfigure}
\end{minipage}
  \caption{(a): Strong wind disturbance case: the reference planner allows to reach the goal and avoid obstacles despite strong wind disturbance. 
  (b): The reference planner switches between the two operation modes, resulting in different velocity constraints and uncertainty bounds.}
  \label{fig:disturbances}
\end{figure}

Summarizing, introducing additional control modes allows avoiding conservative behavior while satisfying constraints and being computationally feasible.

\section{Conclusions and outlook}
\label{sec.Conclusion}
We propose a tightly integrated hierarchical predictive control and planning approach. 
Both planner and controller repeatedly solve moving horizon optimal control problems. The upper- and lower-layers exploit different ``contracts'' (guaranteed uncertainty bounds), i.e., the planner can choose other low-layer controller regions/modes instead of using a fixed safety corridor to capture the controller capabilities. 
For example, the planner can choose slow-speed movement with high precision or a fast speed with large uncertainty bounds. 

The planning reference  is determined by solving a moving horizon optimization problem considering a simplified model. It exploits constraints tightening, which represents the lower-layer tracking capabilities in the form of the precision contracts. The resulting planning algorithm can be reformulated as a MILP, allowing for efficient and reliable solutions. 
The low-layer tube-based MPC controller utilizes a cyclic horizon and results in a  convex optimization problem. It guarantees constraint satisfaction and the desired tracking accuracy for the different modes. Moreover, it operates at a faster time scale. 
We derived conditions that ensure compatibility between the planning and control layers to guarantee recursive feasibility and ensure the satisfaction of constraints and obstacle avoidance.

Simulation results demonstrated the efficiency and applicability of the proposed hierarchical strategy.
First, the contract option provides significant advantages, \eg\  it leads to a less conservative solution. 
Moreover, the hierarchical decomposition of the challenging vehicle control/planning problem leads to a decrease in the computational cost. It allows the implementation of robust control on-board while providing guarantees.

Possible extension are the consideration of  ellipsoidal tube MPC methods \cite{villanueva2017computing,hu2018real}. In this case, the lower controller online sends the tube parameterization to the upper layer. Therefore, the planner can predict a possible uncertainty evaluation over the planning horizon.

We also aim to experimentally evaluate the approach, implementing the upper-layer planner and the lower-layer MPC controller on computationally limited systems.
\bibliography{bib_contract}{}

\appendix
\section{Quadcopter model}\label{sec:ex}
The quadcopter states and inputs are represented in two different coordinate systems, e.g., earth and body fixed frame, see Fig.~\ref{fig:coordinates}. The resulting nonlinear dynamics  are given by \cite{alexis2016robust}:
\begin{align}
\begin{bmatrix}
m_t\mathbf{I}_{3 \times 3} & {0} \\ {0} & \mathbf{I}
\end{bmatrix}
\begin{bmatrix} \dot{{V}} \\ {W} \end{bmatrix}
+
\begin{bmatrix}
{W} \times m_t {V} \\  {W} \times \mathbf{I} {W}
\end{bmatrix}
=
\begin{bmatrix}{F} \\ {T} \end{bmatrix},
\end{align}\label{eq:lin_model_q}
where $m_t$ and ${I}$ are the mass and inertia matrix, ${V}$ and ${W}$ are the linear and angular velocities expressed in the body-fixed frame. ${F}$ and ${T}$ are the applied forces and moments. The model is linearized assuming decoupling of the translational and the attitude dynamics \cite{alexis2016robust}, leading to
\begin{subequations}
	\begin{align}
	\dot{x}_{\text{3D}} &= A_{\text{3D}}x_{\text{3D}}+B_{\text{3D}}u_{\text{3D}},\\
	y_{\text{3D}}  &=C_{\text{3D}}x_{\text{3D}}+D_{\text{3D}}u_{\text{3D}},
	\end{align}
\end{subequations}
with: $
x_{\text{3D}} = [x^\top_{\text{lon}}\, x^\top_{\text{lat}}\, x^\top_{\text{alt}}]^\top,\quad 
u_{\text{3D}} = [u^\top_{\text{lon}}\, u^\top_{\text{lat}}\, u^\top_{\text{alt}}]^\top,$ $
A_{\text{3D}} = \left[\begin{smallmatrix}A_{\text{lon}} & 0 & 0 \\0 & A_{\text{lat}} & 0\\ 0 & 0& A_{\text{alt}} \end{smallmatrix}\right],$ $
B_{\text{3D}} = \left[\begin{smallmatrix}B_{\text{lon}} & 0 & 0 \\ 0 & B_{\text{lat}} & 0\\ 0 & 0 & B_{\text{alt}} \end{smallmatrix}\right].
$ 
The corresponding longitudinal, lateral and vertical sub-dynamics  are given by:
\begin{subequations}
	\begin{align}
	\dot{x}_{\text{lon}} &= A_{\text{lon}}x_{\text{lon}}+B_{\text{lon}}u_{\text{lon}},\\
	\dot{x}_{\text{lat}} &= A_{\text{lat}}x_{\text{lat}}+B_{\text{lat}}u_{\text{lat}},\\
	\dot{x}_{\text{alt}} &= A_{\text{alt}}x_{\text{alt}}+B_{\text{alt}}u_{\text{alt}},
	\end{align}
\end{subequations}
with the matrices
$
A_{\text{lon}}=\left[\begin{smallmatrix}
0 & 1             & 0 & 0\\
0 & -\lambda_{x} & -g& 0\\
0 & 0 & 0 & 1\\
0 & 0 & -a_{w_x,\theta} & -a_{w_x,w_x}
\end{smallmatrix}\right],$ $ B_{\text{lon}}=\left[\begin{smallmatrix} 0\\ 0\\ 0 \\ b_{x}\end{smallmatrix}\right],
$
where longitudinal state is $x_{\text{lon}} = [p_x\, v_x\, \theta\, w_x]^\top$, with the input $u_{\text{lon}} = \theta_c$.
The lateral state is $x_{\text{lat}} = [p_y\, v_y\, \phi\, w_y]^\top$, with the input $u_{\text{lat}} = \phi_c$ and the matrices
$
A_{\text{lat}}=\left[\begin{smallmatrix}
0 & 1 & 0 & 0\\
0 & -\lambda_y & g & 0\\
0 & 0 & 0 & 1\\
0 & 0 & -a_{w_y,\phi} & -a_{w_y,w_y}
\end{smallmatrix}\right],$ $ B_{\text{lat}}=\left[\begin{smallmatrix} 0\\ 0\\ 0 \\ b_{y}\end{smallmatrix}\right]. 
$ 
Finally, the matrices for vertical (altitude) dynamics are given by
$
A_{\text{alt}}=
\left[\begin{smallmatrix}
0 & 1 \\
0 & -\lambda_z 
\end{smallmatrix}\right],$ $ B_{\text{alt}}=
\left[\begin{smallmatrix} 0 \\ b_z\end{smallmatrix}\right],
$
with the vertical state $x_{\text{alt}} = [p_z\, v_z]^\top$, and the input is $u_{\text{alt}} = T_c$.  
The overall UAV states are the roll and pitch angles $(\phi,\theta)$, roll and pitch rates $(w_x,w_y)$, the positions $(p_x,p_y,p_z)$, and the  velocities $(v_x,v_y,v_z)$. 
The  parameters $\lambda_{x}, \lambda_y, \lambda_z, a_{w_x,\theta},  a_{w_x,w_x}, a_{w_y,\phi}, a_{w_y,w_y}, b_{x}, b_{y}, b_{z}$ can be found in  \cite{alexis2016robust}.

\end{document}